\documentclass[sn-chicago,Numbered]{sn-jnl}
\usepackage{graphicx}
\usepackage{multirow}
\usepackage{amsmath,amssymb,amsfonts}
\usepackage{amsthm}
\usepackage{mathrsfs}
\usepackage[title]{appendix}
\usepackage{xcolor}
\usepackage{textcomp}
\usepackage{manyfoot}
\usepackage{booktabs}
\usepackage{algorithm}
\usepackage{algorithmicx}
\usepackage{algpseudocode}
\usepackage{listings}

\theoremstyle{thmstyleone}
\newtheorem{theorem}{Theorem}
\newtheorem{proposition}[theorem]{Proposition}
\theoremstyle{thmstyletwo}

\newtheorem{lemma}[theorem]{Lemma}
\newtheorem{corollary}[theorem]{Corollary} 
\newtheorem{observation}[theorem]{Observation}
\newtheorem{definition}[theorem]{Definition}
\usepackage{multicol}
\usepackage{multirow}
\usepackage{arydshln}
\usepackage{latexsym,stmaryrd,mathtools,amsmath,yhmath} 
\usepackage{graphicx,epsfig,tikz}
\usetikzlibrary{shapes.geometric,positioning,chains,fit,calc, patterns}
\usepackage{makecell, longtable, lscape,lscape, supertabular,adjustbox}
\usepackage{csvsimple}	
\usepackage{supertabular}
\usepackage{makecell}

\newcommand{\ta}{{\tt a} \relax}
\newcommand{\tb}{{\tt b} \relax}

\usepackage{enumitem}
\usepackage[useregional]{datetime2}

\usepackage{comment}

\definecolor{darkred}{RGB}{128,0,0}
\definecolor{darkgreen}{RGB}{0,128,0}
\definecolor{lightgreen}{RGB}{224,255,224}
\definecolor{darkblue}{RGB}{0,0,128}
\definecolor{lightblue}{RGB}{224,244,255}

\newcommand{\vir}[1]{``#1''}

\usepackage[ddmmyyyy]{datetime}
\usepackage{subcaption}

\usepackage{siunitx}

\def\lcp{\mbox{\rm {lcp}}}

\newcommand{\runs}{\textit{runs}}
\newcommand{\bwt}{\textrm{BWT}}
\newcommand{\rev}{\textrm{rev}}
\newcommand{\conj}{\textrm{conj}}
\newcommand{\Oh}{{\cal O}}

\newcommand{\CA}{\textrm{CA}}
\newcommand{\vv}[2]{\beta^{#2}(#1)}

\def\a{{\tt a}}
\def\b{{\tt b}}

\def\x{{\tt x}}

\usepackage{xcolor}
\usepackage{colortbl}
\usepackage{array}
\newcolumntype{P}[1]{>{\centering\arraybackslash}p{#1}}
\usepackage{rotating}
\usepackage{adjustbox}

\begin{document}

\title[Bit catastrophes for the Burrows-Wheeler Transform]{Bit catastrophes for the Burrows-Wheeler Transform}

\author*[1]{\fnm{Sara} \sur{Giuliani}}\email{sara.giuliani\_01@univr.it}

\author*[2]{\fnm{Shunsuke} \sur{Inenaga}}\email{inenaga.shunsuke.380@m.kyushu-u.ac.jp}

\author*[1]{\fnm{Zsuzsanna} \sur{Lipt\'{a}k}}\email{zsuzsanna.liptak@univr.it}

\author*[3]{\fnm{Giuseppe} \sur{Romana}}\email{giuseppe.romana01@unipa.it}

\author*[3]{\fnm{Marinella} \sur{Sciortino}}\email{marinella.sciortino@unipa.it}

\author*[4]{\fnm{Cristian} \sur{Urbina}}\email{crurbina@dcc.uchile.cl}

\affil[1]{\orgname{University of Verona}, \city{Verona}, \country{Italy}}
\affil[2]{\orgname{Kyushu University}, \city{Fukuoka}, \country{Japan}}
\affil[3]{\orgname{University of Palermo}, \city{Palermo}, \country{Italy}}
\affil[4]{\orgname{University of Chile}, \city{Santiago}, \country{Chile}}

\abstract{A bit catastrophe, loosely defined, is when a change in just one character of a string causes a significant change in the size of the compressed string. We study this phenomenon for the Burrows-Wheeler Transform (BWT), a string transform at the heart of several of the most popular compressors and aligners today. The parameter determining the size of the compressed data is the number of equal-letter runs of the BWT, commonly denoted $r$. 

We exhibit infinite families of strings in which insertion, deletion, resp.\ substitution of one character increases $r$ from constant to $\Theta(\log n)$, where $n$ is the length of the string. These strings can be interpreted both as examples for an increase by a multiplicative or an additive $\Theta(\log n)$-factor. As regards multiplicative factor, they attain the upper bound given by Akagi, Funakoshi, and Inenaga [Inf \& Comput. 2023] of $\Oh(\log n \log r)$, since here $r=\Oh(1)$. 

We then give examples of strings in which insertion, deletion, resp.\ substitution of a character increases $r$ by a $\Theta(\sqrt n)$ additive factor. These strings significantly improve the best known lower bound for an additive factor of $\Omega(\log n)$ [Giuliani et al., SOFSEM 2021]. }


\maketitle

\section{Introduction}\label{sec:intro} 
The Burrows-Wheeler Transform (BWT)~\cite{BurrowsWheeler94} is a reversible transformation of a string, consisting of a permutation of its characters. It can be obtained by sorting all of its rotations and then concatenating their last characters. The BWT is present in several compressors, such as {\em bzip}~\cite{bzip}. It also lies at the heart of some of the most powerful compressed indexes in terms of query time and functionality, such as the well-known {\em FM-index}~\cite{FerraginaM00}, and the more recent {\em RLBWT}~\cite{MakinenN05} and {\em r-index}~\cite{GagieNP18,GagieNP20,BannaiGI20,BoucherCL0S21,BoucherCL0S24}. Some of the most commonly used bioinformatics tools such as {\tt bwa}~\cite{bwa}, {\tt bowtie}~\cite{Bowtie}, and SOAP2~\cite{soap2} also use the BWT at their core. 

Given a string $w$, the measure $r=r(w)$ is defined as $r(w)=runs(\bwt(w))$, where $runs(v)$ denotes the number of maximal equal-letter runs of a string $v$. It is well known that $r$ tends to be small on repetitive inputs. This is because, on texts with many repeated substrings, the BWT tends to create long runs of equal characters (so-called {\em clustering effect})~\cite{RosoneS13}. Due to the widespread use of runlength-compressed BWT-based data structures, the BWT can thus be viewed as a {\em compressor}, with the number of runs $r$ the compression size. On the other hand, $r$ is also increasingly being used as a {\em repetitiveness measure}, i.e.\ as a parameter of the input text itself. In his recent survey~\cite{Navarro21a}, Navarro explored the relationships between many repetitiveness measures, among these $r$. In particular, all repetitiveness measures considered are lower bounded by $\delta$~\cite{KociumakaNP20}, a measure closely related to the {\em factor complexity} function~\cite{Loth2}. It was further shown in~\cite{KempaK22} that $r$ is within a $\mathrm{polylog}(n)$ factor of $z$, the number of phrases of the LZ77-compressor~\cite{ZivL77}. 

Giuliani et al.~\cite{GILPST21} studied the ratio of the runs of the BWT of a string and of its reverse. The authors gave an infinite families of strings for which this ratio is $\Theta(\log n)$, where $n$ is the length of the string. This family can also serve as an example of strings in which appending one character can cause $r$ to increase from $\Theta(1)$ to $\Theta(\log n)$. 
In this paper we further explore this effect, extending it to the other edit operations of deletion and substitution, for which we also give examples of a change from $\Theta(1)$ to $\Theta(\log n)$. 
Note that this attains the known upper bound of $\Oh(\log r \log n)$~\cite{AKAGI2023}. 

Akagi et al.~\cite{AKAGI2023} explored the question of how changes of just one character affect 
the compression ratio of known compressors; they refer to this as the compressors' {\em sensitivity}. More precisely, the sensitivity of a compressor is the maximum {\em multiplicative factor} by which a single-character edit operation can change the size of the output of the compressor. In addition, they also study the maximum {\em additive factor} an edit operation may cause in the output. Our second family of strings falls in this category: these are strings with $r$ in $\Theta(\sqrt n)$ on which an edit operation (insertion, deletion, or substitution) can cause $r$ to increase by a further additive factor of $\Theta(\sqrt n)$. This is a significant improvement over the previous lower bound of $\Omega(\log n)$~\cite{GILPST21}. 

Lagarde and Perifel in~\cite{LZbitcatastrophe} show that Lempel-Ziv 78 (LZ78) compression~\cite{ZivL78} suffers from the so-called ``one-bit catastrophe'': they give an infinite family of strings for which prepending a character causes a previously compressible string to become incompressible. They also show that this ``catastrophe'' is not a ``tragedy'': they prove that this can only happen when the original string was already poorly compressible. 

Here we use the term ``one-bit catastrophe'' in a looser meaning, namely simply to denote the effect that an edit operation may change the compression size significantly, i.e.\ increase it such that $r(w_n') = \omega(r(w_n))$, for an infinite family $(w_n)_{n>0}$, where $w_n'$ is the word resulting from applying a single edit operation to $w_n$. For a stricter terminology we would need to decide for one of the different definitions of BWT-compressibility currently in use.
In particular, a string may be called compressible with the BWT if $r$ is in $\Oh(n/\mathrm{polylog}(n))$~\cite{KempaK22}, or if $r(w)/\runs(w) \to\ 0$~\cite{FrosiniMRRS22_DLT}, or even as soon as $\runs(w) > r(w)$~\cite{MANTACI_TCS2017}. 

Some of our bit catastrophes can also be thought of as ``tragedies'', since 
the example families of the first group are precisely those with the best possible compression: their BWT has $2$ runs. In this sense our result on the BWT is even more surprising than that of~\cite{LZbitcatastrophe} on LZ78. 

Note that, in contrast to Lempel-Ziv compression, for the BWT, appending, prepending, and inserting are equivalent operations, since the BWT is invariant w.r.t.\ conjugacy. This means that, if there exists a word $w$ and a character $c$ s.t.\ appending $c$ to $w$ causes a certain change in $r$, then this immediately implies the existence of equivalent examples for prepending and inserting character $c$. This is because $r(wc)/r(w) = x$ (appending) implies that $r(c w)/r(w)=x$ (prepending), as well as $r(uc v)/r(uv) = x$, for every conjugate $uv$ of $w=vu$ (insertion). 

Finally, the BWT comes in two variants: in one, the transform is applied directly on the input string: this is the preferred variant in literature on combinatorics on words, and the one we concentrate on in most of the paper. In the other, the input string is assumed to have an end-of-string marker, usually denoted $\$$: this variant is common in the string algorithms and data structures literature. 
We show that there can be a multiplicative $\Theta(\log n)$, or an additive $\Theta(\sqrt n)$ factor difference between the two transforms. It is interesting to note that the previous remark about the equivalence of insertion in different places in the text does not extend to the variant with the final dollar. We show, however, that our results regarding the $\Theta(\sqrt n)$ additive factor apply also to this variant, for all three edit operations, and that appending a character at the end of the string---i.e., just before the $\$$-character---can result in a multiplicative $\Theta(\log n)$ increase. This is in stark contrast with the known fact that {\em prepending} a character can change the number of runs of the $\$$-variant by at most $2$~\cite{AKAGI2023}.

This work is an extended version of our conference article with the same title, published in the  proceedings of DLT 2023~\cite{GILRSU_DLT}. 

\section{Preliminaries}\label{sec:basics}

In this section, we give the necessary definitions and terminology used throughout the paper. 

\subsection{Basics on words} 

Let $\Sigma$ be a finite ordered {\em alphabet}, of cardinality $\sigma$. The elements of $\Sigma$ are called {\em characters} or {\em letters}. A {\em word} (or {\em string}) over $\Sigma$ is a finite sequence $w=w[0]w[1]\cdots w[n-1]=w[0..n-1]$ of characters from $\Sigma$. We denote by $|w|=n$ the length of $w$, with $\epsilon$ the unique word of length $0$. The set of words of length $n$ is denoted $\Sigma^n$, and $\Sigma^* = \cup_{n\geq 0} \Sigma^n$ is the set of all words over $\Sigma$. Given a word $w=w[0..n-1]$, its {\em reverse} is the word $\rev(s) = w[n-1]w[n-2]\cdots w[0]$. For a non-empty word $w = w[0..n-1]$, we denote by $\widehat{w}$ the word $w[0.. n-2]$, i.e.\ $w$ without its last character. We use the notation $\prod^{k}_{i=1} v_i$ for the concatenation of the words $v_1, v_2,\dots,v_k$. In particular, $v^k$ for $k\geq 1$ stands for the $k$-fold concatenation of the word $v$. 

Let $w$ be a word over $\Sigma$. If $w=uxv$ for some (possibly empty) words $u,x,v\in \Sigma^*$, then $u$ is called a {\em prefix}, $x$ a {\em substring} (or {\em factor}), and $v$ a {\em suffix} of $w$. A {\em proper} prefix, substring, or suffix of $w$ is one that does not equal $w$. If $x$ is a substring of $w$, then there exist $i,j$ such that $x = w[i]\cdots w[j] = w[i..j]$, where $w[i..j]=\epsilon$ if $i>j$. If $w[i..j]=x$, then $i$ is called an {\em occurrence} of $x$. 

Let $u,v \in \Sigma^*$. If $w=uv$ and $w'=vu$, then $w$ and $w'$ are called {\em conjugates} or {\em rotations} of one another. Equivalently, $w'$ is a conjugate of $w$ if there is $0\leq i \leq |w|$ such that $w' = w[i .. |w|]w[0..i-1]$. In this case, we write $w'=\conj_i(w)$.
A word $u$ is a {\em circular factor} of a word $w$ if it is a prefix of $conj_i(w)$ for some $0\leq i <|w|$, and $i$ is called an {\em occurrence} of $u$. If a word $w$ can be written as $w=v^k$ for some $k>1$, then $w$ is called a {\em power}, otherwise $w$ is called {\em primitive}. It is easy to see that $w$ is primitive if and only if it has $|w|$ distinct conjugates.  

Given two words $v,w$, the {\em longest common prefix} of $v$ and $w$, denoted $\lcp(v, w)$, is the unique word $u$ such that $u$ is a prefix of both $v$ and $w$, 
and $v[|u|]\neq w[|u|]$ if neither of the two words is prefix of the other.
The {\em lexicographic order} on $\Sigma^*$ is defined as follows: $v \leq_{\text lex} w$ if $v=\lcp(v,w)$, or else if $v[|u|] < w[|u|]$, where $u=\lcp(v,w)$.
A word is called a \emph{Lyndon word} if it is lexicographically strictly smaller than all of its conjugates. 

Finally, an {\em equal-letter run} (or simply {\em run}) is a maximal substring consisting of the same character, and $\runs(v)$ is the number of equal-letter runs in the word $v$. For example, the word {\tt catastrophic} has $12$ runs, while the word {\tt mississippi} has $8$ runs. 

\subsection{The Burrows-Wheeler Transform} 

Let $w\in \Sigma^*$. The {\em conjugate array} $\CA = \CA_w$ of $w$ is a permutation of $\{0,1,\ldots,n-1\}$ defined by: $\CA[i]<\CA[j]$ if (i) $\conj_i(w) <_{\text{lex}} \conj_j(w)$, or (ii) $\conj_i(w) = \conj_j(w)$ and $i<j$. So $\CA[k]$ contains the index of the $k$th conjugate of $w$ in lexicographic order. (Note that the conjugate array is the circular equivalent of the suffix array.) For example, if $w = $ {\tt catastrophic}, then $\CA_w = [3, 1, 0, 11, 9, 10, 7, 8, 6, 4, 2, 5]$. 

The {\em Burrows-Wheeler Transform} (BWT) of the word $w$ is a permutation of the characters of $w$, usually denoted $L=\bwt(w)$, defined as $L[i]=w[n-1]$ if $\CA[i]=0$, and $L[i] = w[\CA[i]-1]$ otherwise. For example, the BWT of the word {\tt catastrophic} is {\tt tcciphrotaas}. It follows from the definition that $w$ and $w'$ are conjugates if and only if $\bwt(w)=\bwt(w')$. 

We denote with $r(w) = \runs(\bwt(w))$ the number of runs in the BWT of the word $w$. For example, $r({\tt catastrophic}) = \runs({\tt tcciphrotaas}) = 10$. 

In the context of string algorithms and data structures, it is usually assumed that each string terminates with an end-of-string symbol (denoted by $\$$), not occurring elsewhere in the string; the $\$$-symbol is smaller than all other symbols in the alphabet. In fact, with this assumption, sorting the conjugates of $w\$$ can be reduced to lexicographically sorting its suffixes. Note that appending the character $\$$ to the word $w$ changes the output of BWT. We denote by $r_{\$}(w)=runs(\bwt(w\$))$. For example, $\bwt({\tt catastrophic\$})={\tt ctci\$phrotaas}$ and $r_{\$}({\tt catastrophic})=12$.

\subsection{Standard words}

Given an infinite sequence of integers $(d_0 , d_1 , d_2 , \ldots)$, with $d_0 \geq 0,$ and $d_i > 0$ for all $i > 0$, called a {\em directive sequence}, define a sequence of words $s_i$ with $i \geq 0$ of increasing length as follows: $s_0 = \b$, $s_1 = \a$, $s_{i+1} = s_i^{d_{i-1}} s_{i-1}$, for $i \geq 1$. The words $s_i$ are called {\em standard words}. The index $i$ is referred to as the {\em order} of $s_i$. 
Without loss of generality, here we can assume that $d_0>0$ (otherwise, the role of $\b$ and $\a$ is exchanged.). 
It is known that for $i\geq 2$, 
every standard word $s_i$ can be written as $s_i = x_i \a\b$ if $i$ is even, $s_i = x_i \b\a$ if $i$ is odd, where the factor $x_i$ is a palindrome \cite{deLuca97a}. 

Standard words are a well studied family of binary words with a lot of interesting combinatorial properties and characterizations and appear as extreme cases in many contexts~\cite{KMP77, CS2010, SciortinoZ07, MantaciRS03}. In particular, in \cite{MantaciRS03}, it was shown that the BWT of a binary word has exactly two runs if and only if it is a conjugate of a standard word or a conjugate of a power of a standard word.

{\em Fibonacci words} are a particular case of standard words, with directive sequence consisting of only ones. More precisely, Fibonacci words can be defined as follows: $s_0 = {\tt b}$, $s_1 = {\tt a}$, $s_{i+1} = s_i s_{i-1}$, for $i \geq 1$. It is easy to see that for all $i\geq 0$, $|s_i|=F_i$, where $(F_i)_{i\geq 0}$ is the Fibonacci sequence $1,1,2,3,5,8,13,21,\ldots$, defined by the recurrence $F_0=F_1=1$, and $F_{i+1} = F_i+F_{i-1}$ for $i\geq 1$. Using the well-known fact that the Fibonacci sequence grows exponentially in $i$, we have that $i = \Theta(\log |s_i|)$.
Moreover, for all $k\geq 1$, $s_{2k}=x_{2k}{\tt ab}$ and $s_{2k+1}=x_{2k+1}{\tt ba}$, where $x_{2k}$ and $x_{2k+1}$ are  palindromes (in particular, $x_2=\epsilon$). These words $x_h,$ for $h\geq 2$, are also referred to as {\em central words}. The recursive structure of the words $x_{2k}$ and $x_{2k+1}$ is also known \cite{deLucaM94}: $x_{2k}=x_{2k-1}{\tt ba}x_{2k-2}=x_{2k-2}{\tt ab}x_{2k-1}$ and $x_{2k+1}=x_{2k}{\tt ab}x_{2k-1}=x_{2k-1}{\tt ba}x_{2k}$.

\section{Increasing $r$ by a $\Theta(\log n)$-factor}\label{sec:Fibonacci}

In this section we give infinite families of words on which a single edit operation, such as insertion, deletion or substitution of a character, can cause an increase of $r$ from constant to $\Theta(\log n)$, where $n$ is the length of the word. 
The impact of the three edit operations on the BWT of the word is shown in Figure \ref{table:ins_del_sub}.

\begin{figure}
		 \centering
		\begin{subfigure}[b]{0.4\textwidth}
            \resizebox{.85\textwidth}{!}{
			\centering
			\begin{tabular}{r | r c c |}		
				  & CA &  \multicolumn{1}{p{3cm}}{\centering rotations of {\tt abaababaabaab}} & \multicolumn{1}{p{0.1cm}}{\centering } \\
				\hline
				0  &  7  & {\tt aabaababaabab } & {\tt b }\\
				1  &  2  & {\tt aababaabaabab } & {\tt b }\\
				2  &  10  & {\tt aababaababaab } & {\tt b }\\
				3  &  5  & {\tt abaabaababaab } & {\tt b }\\
				4  &  0  & {\tt abaababaabaab } & {\tt b }\\
				5  &  8  & {\tt abaababaababa } & {\tt a }\\
				6  &  3  & {\tt ababaabaababa } & {\tt a }\\
				7  &  11  & {\tt ababaababaaba } & {\tt a }\\
				8  &  6  & {\tt baabaababaaba } & {\tt a }\\
				9  &  1  & {\tt baababaabaaba } & {\tt a }\\
				10  &  9  & {\tt baababaababaa } & {\tt a }\\
				11  &  4  & {\tt babaabaababaa } & {\tt a }\\
				12  &  12  & {\tt babaababaabaa } & {\tt a }\\
			\end{tabular}
                }
			\caption[]
			{Fibonacci word of order 6}
			\label{fig:fib6}
		\end{subfigure}
		\qquad
		\begin{subfigure}[b]{0.4\textwidth}  
                \resizebox{.85\textwidth}{!}{
			\centering 
			\begin{tabular}{r | r c c |}	
					& CA &  \multicolumn{1}{p{3cm}}{\centering rotations of {\tt abaaba\textcolor{red}{b}baabaab}} & \multicolumn{1}{p{0.5cm}}{\centering } \\
				\hline
				0 & 8  & {\tt aabaababaaba\textcolor{red}{b}b } & {\tt b }\\
    			1 & 11  & {\tt aababaaba\textcolor{red}{b}baab } & {\tt b }\\
    			2 & 2  & {\tt aaba\textcolor{red}{b}baabaabab } & {\tt b }\\
    			3 & 9  & {\tt abaababaaba\textcolor{red}{b}ba } & {\tt a }\\
    			4 & 0  & {\tt abaaba\textcolor{red}{b}baabaab } & {\tt b }\\
    			5 & 12  & {\tt ababaaba\textcolor{red}{b}baaba } & {\tt a }\\
    			6 & 3  & {\tt aba\textcolor{red}{b}baabaababa } & {\tt a }\\
    			7 & 5  & {\tt a\textcolor{red}{b}baabaababaab } & {\tt b }\\
    			8 & 7  & {\tt baabaababaaba\textcolor{red}{b} } & {\tt \textcolor{red}{b} }\\
    			9 & 10  & {\tt baababaaba\textcolor{red}{b}baa } & {\tt a }\\
    			10 & 1  & {\tt baaba\textcolor{red}{b}baabaaba } & {\tt a }\\
    			11 & 13  & {\tt babaaba\textcolor{red}{b}baabaa } & {\tt a }\\
    			12 & 4  & {\tt ba\textcolor{red}{b}baabaababaa } & {\tt a }\\
    			13 & 6  & {\tt \textcolor{red}{b}baabaababaaba } & {\tt a }\\
								
			\end{tabular}
                }
			\caption[]
			{Insertion}    
			\label{fig:fibins}
		\end{subfigure} 

            \medskip
  
		\begin{subfigure}[b]{0.4\textwidth}  
                \resizebox{.85\textwidth}{!}{
			\centering 
			\begin{tabular}{r | r c c |}		
			& CA &  \multicolumn{1}{p{3cm}}{\centering rotations of {\tt abaababaabaa}} & \multicolumn{1}{p{0.5cm}}{\centering } \\
			\hline
			0 & 10  & {\tt aaabaababaab } & {\tt b }\\
			1 & 7  & {\tt aabaaabaabab } & {\tt b }\\
			2 & 11  & {\tt aabaababaaba } & {\tt a }\\
			3 & 2  & {\tt aababaabaaab } & {\tt b }\\
			4 & 8  & {\tt abaaabaababa } & {\tt a }\\
			5 & 5  & {\tt abaabaaabaab } & {\tt b }\\
			6 & 0  & {\tt abaababaabaa } & {\tt a }\\
			7 & 3  & {\tt ababaabaaaba } & {\tt a }\\
			8 & 9  & {\tt baaabaababaa } & {\tt a }\\
			9 & 6  & {\tt baabaaabaaba } & {\tt a }\\
			10 & 1  & {\tt baababaabaaa } & {\tt a }\\
			11 & 4  & {\tt babaabaaabaa } & {\tt a }\\
			&&&\\			
			\end{tabular}
               }
			\caption[]
			{Deletion}    
			\label{fig:fibdel}
		\end{subfigure}
		\qquad
		\begin{subfigure}[b]{0.4\textwidth}   
                \resizebox{.85\textwidth}{!}{
			\centering 
			\begin{tabular}{r | r c c |}		
				& CA &  \multicolumn{1}{p{3cm}}{\centering rotations of {\tt abaababaabaa\textcolor{red}{a}}} & \multicolumn{1}{p{0.5cm}}{\centering } \\
				\hline
				0 & 10  & {\tt aa\textcolor{red}{a}abaababaab } & {\tt b }\\
    			1 & 11  & {\tt a\textcolor{red}{a}abaababaaba } & {\tt a }\\
    			2 & 7  & {\tt aabaa\textcolor{red}{a}abaabab } & {\tt b }\\
    			3 & 12  & {\tt \textcolor{red}{a}abaababaabaa } & {\tt a }\\
    			4 & 2  & {\tt aababaabaa\textcolor{red}{a}ab } & {\tt b }\\
    			5 & 8  & {\tt abaa\textcolor{red}{a}abaababa } & {\tt a }\\
    			6 & 5  & {\tt abaabaa\textcolor{red}{a}abaab } & {\tt b }\\
    			7 & 0  & {\tt abaababaabaa\textcolor{red}{a} } & {\tt \textcolor{red}{a} }\\
    			8 & 3  & {\tt ababaabaa\textcolor{red}{a}aba } & {\tt a }\\
    			9 & 9  & {\tt baa\textcolor{red}{a}abaababaa } & {\tt a }\\
    			10 & 6  & {\tt baabaa\textcolor{red}{a}abaaba } & {\tt a }\\
    			11 & 1  & {\tt baababaabaa\textcolor{red}{a}a } & {\tt a }\\
    			12 & 4  & {\tt babaabaa\textcolor{red}{a}abaa } & {\tt a }\\ 	
			\end{tabular}
                }
			\caption[]
			{Substitution}
			\label{fig:fibsub}
		\end{subfigure}
\caption{The BWT matrix of the Fibonacci word of order 6 (a), and that of the result for 3 bit-catastrophes:  (b) inserting a character in position $6 = F_{6-1}-2$,  (c) deleting the last character, (d) substituting the last character.}
  \label{table:ins_del_sub}
\end{figure}

\subsection{Inserting a character}\label{subsec:fibinsertion}

First we recall a result from~\cite{GILPST21}, namely that appending a character to the reverse of a Fibonacci word can increase the number of runs by a logarithmic factor \cite{GILPST21}. This result was shown using the following proposition: 

\bigskip

\begin{proposition}[\cite{GILPST21}, Prop.~3]
\label{prop:furtherC_at_the_end}
Let $v$ be the reverse of a Fibonacci word $s$. If $s$ is of even order $2k$, then $r(v \b) = 2k$. If $s$ is of odd order $2k+1$, then $r(v \a) = 2k$.
\end{proposition}

\bigskip

\noindent
A well-known property of each standard word is that its reverse is one of its rotations~\cite{deLuca97a}.
Since $s$ is a Fibonacci word, and thus a standard word, its reverse $v$ has the same BWT as $s$.Since $s$ is a standard word, $r(s)=2$, and therefore, also $r(v)=2$. Using the fact that the length of the $i$th Fibonacci word is $F_i$, and that the Fibonacci sequence $(F_i)_{i\geq 0}$ grows exponentially in $i$, it follows that by appending a final $\b$, the number of runs of the BWT is increased by a logarithmic factor in $n=|v|$, namely from $2=\Oh(1)$ to either $i$ (if $i=2k$, for some $k$) or $i-1$ (if $i=2k+1$, for some $k$), which are both $\Theta(\log n)$. 
\medskip

Similarly to Prop.~\ref{prop:furtherC_at_the_end}, we will prove that adding a character $\x$ greater than $\b$ and not present in the word has the same effect as adding a further $\b$ at the end of the reverse of a Fibonacci word of even order. Intuitively, this is because in both cases a new factor is introduced in the word, namely $\b\b$ respectively $\x$. Both these factors are greater than all the other factors of the word, and they are the only change in the word. Adding a further $\a$ to the reverse of a Fibonacci word of odd order, or a character smaller than $\a$, have a similar effect. We formalize this in the following proposition:

\bigskip

\begin{proposition}\label{prop:furtherC_notinSigma_even}\label{prop:furtherC_notinSigma_odd}
Let $v$ be the reverse of the Fibonacci word $s$, with $|v|\geq 2$. %
\begin{enumerate}
    \item If $s$ is of even order $2k$ and $\x > \b$, then $r(v\x) = 2k+1$. \newline
    \item If $s$ is of odd order $2k+1$ and $\x < \a$, then $2k+2\leq r(v\x) \leq 2k+3$.
\end{enumerate}
\end{proposition}

\begin{proof}
Recall that $\bwt(v\x) = \bwt(\x v)$. For the proof we consider the conjugate $\x v$. 

Let us consider firstly the case in which the Fibonacci word $s$ is of even order $2k$. 
Let us consider the word $v'$ obtained by prepending the character $\b$ to $v$, i.e.  $v' = \b v$. 
If we denote by $n=|v'|$, then $v = v'[1..n-1]$, and $\x v$ and $v'$ differs only for the characters in position $0$. Moreover, $\x v[0] = \x \b$ is the lexicographically greatest 2-length substring of $v$, and $v'[0..1] = \b\b$ is the lexicographically greatest 2-length substring of $v'$. In particular, the relative lexicographic order of the conjugates $\conj_{h}(\x v)$ and $ \conj_{h}(v')$ with $0\leq h\leq  n-1$ is the same.  Therefore the BWT of $\x v$ and $v'$ differs only by the character preceding the conjugates starting in position $1$. By \cite[Proposition 4]{GILPST21}, the character preceding $\conj_{1}(v')$ is the last character of a run of $\b$'s, therefore the character $\x$ which precedes $\conj_{1}(\x v)$ adds only one run, namely the 1-length run of $\x$. Since the $r(v') = 2k$, then $\bwt(\x v)$ has the same $2k$ runs, plus the further run consisting of the single $\x$. Therefore $r(v\x) = r(\x v) =2k+1$. 

Let us consider now the case in which $s$ is of odd order $2k+1$. Let us consider the word $v' = \a v$ of length $n$. If $k=1$, it is easy to verify that $r(\a\b\a\a)=2$ and $r(\a\b\a\x)=4$, so $r(v\x)=2k+2$.
Let us suppose $k>1$. From \cite[Prop. 3]{GILPST21}, we know that $\bwt(\a v) = \b^{F_{2k-2}}\a\a\b^{F_{2k-4}} \cdots \b^{F_2}\a\b^{F_0}\a^{F_{2k}-k+1}$. Let us consider also the conjugate $\x v$ having the same $\bwt$ as $v\x$. Since $v = v'[1..n-1]$, then $\x v$ and $v'$ differs only for the characters in position $0$.
Note that $\x\a$ and $\a\a\a$ are the lexicographically smallest substrings of $\x v$ and $v'$, respectively. It follows that the relative lexicographic order of the conjugates $\conj_{h}(\x v)$ $\conj_{h}(v')$ with $1\leq h\leq n-1$ is the same.
The additional rotation $\x v$ starts with $\x$, and it is now the smallest among all rotations, therefore it will be at the very beginning of the matrix. 
Since $\x v$ ends with $\a$ and the lexicographically following rotations end with $\b$, it increments the number of runs by 1 with respect to the BWT of $v'$. On the other hand, the rotation $\conj_{0}(v')$ ends with $\a$ and follows all the rotations that start with $\a\a$ and end with $\b$, and precedes the rotation $\conj_{n-3}(v')$ that starts with $\a\b\a\a\a$ and ends with $\a$.
Additionally, the rotation ending with $\x$ contributes to $r$ with at most 2 more runs. 
This is because it either falls in between runs of two distinct characters, or within a run of a single character.
In total, $\bwt(\x v)$ has at most $2k+3$ runs, and $2k+2\leq r(\x v)= r(v\x) \leq 2k+3$.
 \end{proof}

 \begin{figure}[h]
	\centering
	\scalebox{.65}{
	\tikzset{every picture/.style={line width=0.75pt}} 

	\subfloat[Insert a $\tt{b} $ in position $F_{2k-1}-2$. \label{fig:ins_b}]
	{\parbox[b][5cm][t]{.45\textwidth}{	\begin{tikzpicture}[x=1pt,y=1pt,yscale=-1,xscale=1]
				
				\draw [color={rgb, 255:red, 126; green, 211; blue, 33 }  ,draw opacity=1 ][line width=3]    (150,26) -- (170,26) ;
				\draw [color={rgb, 255:red, 245; green, 166; blue, 35 }  ,draw opacity=1 ][line width=3]    (70,26) -- (132,26) ;
				\draw    (70,50) -- (150,50) ;
				\draw    (70,40) -- (70,50) ;
				\draw    (150,40) -- (150,50) ;
				\draw    (150,50) -- (190,50) ;
				\draw    (150,40) -- (150,50) ;
				\draw    (190,40) -- (190,50) ;
				
				\draw [color={rgb, 255:red, 245; green, 166; blue, 35 }  ,draw opacity=1 ][line width=3]    (130,85) -- (192,85) ;
				\draw [color={rgb, 255:red, 126; green, 211; blue, 33 }  ,draw opacity=1 ][line width=3]    (93,85) -- (113,85) ;
				
				\draw [color={rgb, 255:red, 245; green, 166; blue, 35 }  ,draw opacity=1 ][line width=3]    (130,125) -- (192,125) ;
				\draw [color={rgb, 255:red, 126; green, 211; blue, 33 }  ,draw opacity=1 ][line width=3]    (93,125) -- (113,125) ;

				\draw (49,22) node [anchor=north west][inner sep=0.75pt]   [align=left] {$s =$};
				\draw (133,22) node [anchor=north west][inner sep=0.75pt]   [align=left] {$\tt{b} \tt{a} $};
				\draw (171,22) node [anchor=north west][inner sep=0.75pt]   [align=left] {$\tt{a} \tt{b} $};			
				\draw (98,52) node [anchor=north west][inner sep=0.75pt]   [align=left] {$F_{2k-1}$};
				\draw (157,52) node [anchor=north west][inner sep=0.75pt]   [align=left] {$F_{2k-2}$};				
				\draw (36,81) node [anchor=north west][inner sep=0.75pt]   [align=left] {$rev(s) =$};
				\draw (75,81) node [anchor=north west][inner sep=0.75pt]   [align=left] {$\tt{b} \tt{a} $};
				\draw (116,81) node [anchor=north west][inner sep=0.75pt]   [align=left] {$\tt{a} \tt{b} $};
				\draw (30,121) node [anchor=north west][inner sep=0.75pt]   [align=left] {$\b  rev(s) =$};
				\draw (70,121) node [anchor=north west][inner sep=0.75pt]   [align=left] {$\tt{b} \tt{b} \tt{a} $};
				\draw (116,121) node [anchor=north west][inner sep=0.75pt]   [align=left] {$\tt{a} \tt{b} $};

			\end{tikzpicture}
	}}
	\qquad
	\subfloat[Insert an $\tt{a} $ in position $F_{2k}-2$\label{fig:ins_a}]
	{\parbox[b][5cm][t]{.45\textwidth}{
			\begin{tikzpicture}[x=1pt,y=1pt,yscale=-1,xscale=1]
				
				\draw [color={rgb, 255:red, 126; green, 211; blue, 33 }  ,draw opacity=1 ][line width=3]    (150,26) -- (170,26) ;
				\draw [color={rgb, 255:red, 245; green, 166; blue, 35 }  ,draw opacity=1 ][line width=3]    (70,26) -- (132,26) ;
				\draw    (70,50) -- (150,50) ;
				\draw    (70,40) -- (70,50) ;
				\draw    (150,40) -- (150,50) ;
				\draw    (150,50) -- (190,50) ;
				\draw    (150,40) -- (150,50) ;
				\draw    (190,40) -- (190,50) ;
				\draw (98,52) node [anchor=north west][inner sep=0.75pt]   [align=left] {$F_{2k-1}$};
				\draw (157,52) node [anchor=north west][inner sep=0.75pt]   [align=left] {$F_{2k-2}$};

				\draw [color={rgb, 255:red, 245; green, 166; blue, 35 }  ,draw opacity=1 ][line width=3]    (112,85) -- (172,85) ;
				\draw [color={rgb, 255:red, 126; green, 211; blue, 33 }  ,draw opacity=1 ][line width=3]    (72,85) -- (92,85) ;
				
				\draw [color={rgb, 255:red, 126; green, 211; blue, 33 }  ,draw opacity=1 ][line width=3]    (170,120) -- (190,120) ;
				\draw [color={rgb, 255:red, 245; green, 166; blue, 35 }  ,draw opacity=1 ][line width=3]    (90,120) -- (152,120) ;
				
				\draw [color={rgb, 255:red, 126; green, 211; blue, 33 }  ,draw opacity=1 ][line width=3]    (170,155) -- (190,155) ;
				\draw [color={rgb, 255:red, 245; green, 166; blue, 35 }  ,draw opacity=1 ][line width=3]    (90,155) -- (152,155) ;
				
				\draw (49,22) node [anchor=north west][inner sep=0.75pt]   [align=left] {$s =$};
				\draw (133,22) node [anchor=north west][inner sep=0.75pt]   [align=left] {$\tt{b} \tt{a} $};
				\draw (171,22) node [anchor=north west][inner sep=0.75pt]   [align=left] {$\tt{a} \tt{b} $};
				
				\draw (49,81) node [anchor=north west][inner sep=0.75pt]   [align=left] {$t =$};
				\draw (176,81) node [anchor=north west][inner sep=0.75pt]   [align=left] {$\tt{b} \tt{a} $};
				\draw (96,81) node [anchor=north west][inner sep=0.75pt]   [align=left] {$\tt{a} \tt{b} $};
				\draw (35,116) node [anchor=north west][inner sep=0.75pt]   [align=left] {$rev(t) =$};
				\draw (76,116) node [anchor=north west][inner sep=0.75pt]   [align=left] {$\tt{a} \tt{b} $};
				\draw (153,116) node [anchor=north west][inner sep=0.75pt]   [align=left] {$\tt{b} \tt{a} $};
				\draw (28,149) node [anchor=north west][inner sep=0.75pt]   [align=left] {$\a  rev(t) =$};
				\draw (70,149) node [anchor=north west][inner sep=0.75pt]   [align=left] {$\tt{a} \tt{a} \tt{b} $};
				\draw (153,149) node [anchor=north west][inner sep=0.75pt]   [align=left] {$\tt{b} \tt{a} $};

			\end{tikzpicture}
			
	}}
}	\caption{Example of adding a $\b$ and an $\a$ within the Fibonacci word of even order at positions $F_{2k-1}-2$ and $F_{2k}-2$, respectively. It causes a logarithmic increment of the number of $\bwt$-runs.}\label{fig:fib_increment}
\end{figure}

 The following proposition can be deduced from Prop.~\ref{prop:furtherC_at_the_end} and shows that there exist at least two positions in a Fibonacci word of even order where adding a character causes a logarithmic increment of $r$. In Fig. \ref{fig:fib_increment} these two positions are shown.

\bigskip
\begin{proposition}\label{prop:insert_infib}
Let $s$ be the Fibonacci word of even order $2k$, and $n=|s|$. Let $s'$ be the word resulting from inserting a $\b$ at position $F_{2k-1}-2$, and $s''$ the word resulting from inserting an $\a$ at position $F_{2k}-2$. Then $\bwt(s')$ and $\bwt(s'')$ have $\Theta(\log n)$ runs. 
\end{proposition}

\begin{proof}
It is known that each standard word and its reverse are conjugate. Let us consider $s=x_{2k} \a \b$, where $x_{2k}$ is palindrome. Moreover, from known properties on Fibonacci words, $x_{2k}=x_{2k-1}\b \a x_{2k-2}=x_{2k-2} \a \b x_{2k-1}$ with $x_{2k-1}$ and $x_{2k-2}$ palindrome, as well. Let $v= \b \a x_{2k}$ be the reverse of $s$. From Proposition \ref{prop:furtherC_at_the_end}, we know that $r(v\b)=2k$. One can verify that appending a $\b$ to $v$
is equivalent to inserting $\b$ at position $F_{2k-1}-2$ of $s$.  On the other hand, the word $t=x_{2k}\b\a$ is a conjugate of both $v$ and $s$. Moreover, by using properties of directive sequences, $t$ is a standard word of odd order $2k-1$ \cite{BersteldeLuca97}. It is easy to verify that $\rev(t)=\conj_{F_{2k}-2}(s)$. By~\cite[Proposition 8]{GILPST21},  $r(\rev(t)\a)=2k-2=\Theta(\log n)$, and for similar considerations as above, appending an $\a$ to $\rev(t)$ is equivalent to inserting $\a$ at position $F_{2k}-2$ in $s$. 
\end{proof}

An analogous result to Prop.~\ref{prop:insert_infib} can be proved for Fibonacci words of odd order. 

\subsection{Deleting a character}\label{subsec:fibdeletion}
We next show that deleting a character can result in a logarithmic increment in $r$. In particular, we consider a Fibonacci word of even order and compute the form of its BWT, as shown in the following proposition.

\bigskip
\begin{proposition}\label{prop:removedC_bwtForm}
Let $s$ be the Fibonacci word of even order $2k>4$ and $\widehat{s} = s[0..n-2]$, where $n=|s|$. Then $\bwt(\widehat{s})$ has the following form:  $$\bwt(\widehat{s}) = \b^{k-1}\a\b^{F_{2k-3}-k+1}\a\b^{F_{2k-5}} \cdots \b^{F_5}\a\b^{F_3}\a\b\a^{F_{2k-1}-k+1}.$$ 
Therefore, $\bwt(\widehat{s})$ has $2k$ runs.
\end{proposition}

\bigskip

To give the proof, we divide the BWT matrix of the word $\widehat{s}$ in three parts: {\em top}, {\em middle} and {\em bottom part}, showing the form of each part separately:
\begin{align*}
    \bwt_{\text{top}}(\widehat{s}) = & \b^{k-1}\a\b^{F_{2k-3}-k+1} \text{, consisting of 3 runs,}\\
    \bwt_{\text{mid}}(\widehat{s}) = & \a\b^{F_{2k-5}}\a\b^{F_{2k-7}}\cdots\a\b \text{, consisting of } 2(k-2) \text{ runs,}\\
    \bwt_{\text{bot}}(\widehat{s}) = & \a^{F_{2k-1}-k+1} \text{, consisting of 1 run}.
\end{align*}

Altogether, we then have $3+2(k-2)+1 = 2k$ runs.

In order to describe the structure of the matrix, we start with the following lemma that provides information on the structure of $s$.

\bigskip
\begin{lemma}\label{le:structure_s}
    Let $s$ be the Fibonacci word of even order $2k>4$ and  $n=|s|$. Then, $s$ can be factorized as follows:
    $$s=x_{2k-1}\b\a x_{2k-3}\b\a\cdots x_7\b\a x_{5}\b\a s_4=x_{2k-2}\a\b x_{2k-3}\b\a x_{2k-4}\a\b \cdots x_4\a\b x_{3}\b\a s_4,$$
where $x_i$ denotes the central word of order $i$, with $3\leq i\leq 2k-1$ and $s_4=x_4\a\b=\a\b\a\a\b$ is the Fibonacci word of order $4$.
\end{lemma}

\begin{proof}
It follows by using the induction on $k$ and the fact that $s=s_{2k-1}s_{2k-2}$ and $s=x_{2k}\a\b$, where  $x_{2k}=x_{2k-1}{\tt ba}x_{2k-2}=x_{2k-2}{\tt ab}x_{2k-1}$ and $x_{2k-1}=x_{2k-3}{\tt ba}x_{2k-2}=x_{2k-2}{\tt ab}x_{2k-3}$.
In fact, the equality $s_6=x_5\b\a x_4 \a\b$ holds since $s_4=x_4 \a\b$. On the other hand, $s_6=x_4\a \b x_5\a\b=x_4 \a\b x_3 \b \a x_4 \a \b$. Since $s_{2k+2}=s_{2k+1}s_{2k}$ and $s_{2k+1}=x_{2k+1}\b\a$, we have that 
$s_{2k+2}=x_{2k+1}\b\a x_{2k-1}\b\a x_{2k-3}\b\a\cdots x_7\b\a x_{5}\b\a s_4b$. On the other hand, $s_{2k+2}=s_{2k+1}x_{2k-2}\a\b x_{2k-3}\b\a x_{2k-4}\a\b \cdots x_4\a\b x_{3}\b\a s_4$. The thesis follows from the fact that $s_{2k+1}=x_{2k+1}\b\a=x_{2k}{\tt ab}x_{2k-1}\b\a$.
\end{proof}
 We identify the following 3 conjugates of the word $\widehat{s}=s[0..n-2]$ of length $n-1$ that delimit the 3 parts of the BWT matrix of the word: 
\begin{align*}
\conj_{n-3}(\widehat{s}) =& \a \a x_{2k-1}\b\a x_{2k-3}\b\a \cdots x_5\b\a \a\b,\\ 
\conj_{n-5}(\widehat{s}) =& x_4 \a x_{2k-1} \cdots x_5 \b\a\\ 
\conj_{0}(\widehat{s}) = &x_{2k}\a
\end{align*}
The structure of these 3 conjugates follows from Lemma \ref{le:structure_s}. It is easy to see that $\conj_{n-3}(\widehat{s}) < \conj_{n-5}(\widehat{s}) < \conj_{0}(\widehat{s})$.
The rotation $\conj_{n-3}(\widehat{s})$ , starting with $\a\a x_{2k-1}$
is the smallest rotation in the matrix due to the unique $\a\a\a$ prefix. The rotation $\conj_{n-5}(\widehat{s})$, starting with $x_4=\a\b\a$ indicates the beginning of the middle part, and it is the smallest rotation starting with $\a\b$. 
Finally, the word itself $\widehat{s} = x_{2k}\a$ determines the beginning of the bottom part, namely the last long run of $\a$'s in the BWT.

The top part of the matrix consists of all rotations of the word starting with $\a\a$.  We give first the following lemma characterizing all occurrences of the factor $\a\a$ in $\widehat{s}$.

\bigskip 
\begin{lemma}\label{le:x_h_even}
    Let $\widehat{s}=s[0..n-2]$, where $s$ is the Fibonacci word of order $2k>4$ and $n=|s|$.  Every occurrence of the circular factor $\a\a$ in $\widehat{s}$ is an occurrence of $\a x_{2h}$ for some $2\leq h\leq k$. For every $2\leq h\leq k-1$ there is exactly one occurrence of $\a x_{2h}$ followed by $\a\a$.
\end{lemma}

\begin{proof}
By using Lemma \ref{le:structure_s}. we have that $\widehat{s}=x_{2k-1}\b\a x_{2k-3}\b\a\cdots x_7\b\a x_{5}\b\a x_4\a$.  It follows from the recursive construction of s that occurrences of $\a\a$ are generated whenever we create $s_{2h} = s_{2h-1}s_{2h-2}=x_{2h-1}\b\a x_{2h-2}\a\b$. This is because each central word $x_i$ starts with $\a\b$ and ends with $\b\a$. Therefore, it follows from the structure of $\widehat{s}$ that we can partition all its occurrences of $\a\a$ into four disjoint sets: the occurrence of $\a x_{2k}$, the occurrences of $\a x_{2h}$ followed by $\a\b$, the ones followed by $\b\a$, the ones followed by $\a\a$, with $4\leq h\leq k-1$. Since $x_{2h}$ ends with $\a$ and the factor $\a\a\a$ occurs only at position $n-3$ in $\widehat{s}$, we can state that the last set of occurrences corresponds to the occurrences of $\a x_{2h}\a\a$ (with $h=2,\ldots,k-1$). We have only one occurrence for each of these factors. In fact, $\a x_4\a\a$ occurs at position $n-6$. Moreover, for $h=3,\ldots,k-1$, every $\a x_{2h-1}$ in the previous factorization of $\widehat{s}$ is an occurrence of $\a \widehat{s_{2h}}$, then it is an occurrence of the circular factor $\a x_{2h}\a\a$. 
\end{proof}

We are now going to show that only one of the rotations of $\widehat{s}$ starting with $\a\a$ ends with an $\a$, and we show where the $\a$ in the BWT of the top part breaks the run of $\b$'s. 

\bigskip

\begin{lemma}[Top part]\label{lemma:top}
Given $\widehat{s}=s[0..n-2]$, where $s$ is the Fibonacci word of order $2k$ and $n=|s|$, then the first $k$ rotations in the BWT matrix are $\a\a\a \cdots \b < \a x_4\a\a \cdots \b < \a x_{6}\a\a \cdots \b < \ldots < \a x_{2k-2}\a\a \cdots \b < \a x_{2k}$. All other $F_{2k-3}-k+1$ rotations starting with $\a\a$ end with a \b.

\end{lemma}			

\begin{proof}
There are $F_{2k-3}+1$ occurrences of $\a\a$. In fact,  $\widehat{s}$ has $F_{2k-1}$ occurrences of $\a$'s and $F_{2k-2}-1$ occurrences of $\b$'s. Since $\b\b$ does not occur in $\widehat{s}$, it follows that $F_{2k-2}-1$ $\a$'s are followed by a $\b$. Therefore there are $F_{2k-1}-F_{2k-2}+1 = F_{2k-3}+1$ occurrences of $\a$ followed by an $\a$.

    Among all rotations starting by $\a\a$ the $k$ smallest ones are those starting with $\a x_h \a \a$ for each even $h$ in increasing order of $h$. 
        This is because of the single occurrence of $\a\a\a$, consisting in the rightmost occurrence of $x_4$ followed by $\a\a$.
	Finally, only $\a x_2k$ is preceded by an $\a$, therefore the $k$ smallest rotations of $\widehat{s}$ are all preceded by a $\b$ except for the largest of them which is preceded by an $\a$. This shows that the $k$ smallest rotations in the BWT matrix form two runs: $\b^{k-1} \a$.
	All the remaining $F_{2k-3}-k+1$ rotations starting with $\a\a$ correspond to some occurrence of $\a x_{2h}$ for some $2\leq h<k$, by using Lemma \ref{le:x_h_even}. By using the properties of central words, each of these occurrences is followed by $\b\a$. Then the correspondent rotation is lexicographically greater than $\a x_{2k}$, which is prefixed by $\a x_{2h}\a\b$. Since there is a unique occurrence of $\a\a\a$, all these rotations starting with $\a\a$ are preceded by $\b$ by construction. Therefore, we have $\b^{k-1}\a\b^{F_{2k-3}-k+1}$ in the top part of the BWT matrix of $\widehat{s}$. 
\end{proof} 

\begin{figure}[h]
	\centering
	\scalebox{.75}{
	\tikzset{every picture/.style={line width=0.75pt}} 

\begin{tikzpicture}[x=0.75pt,y=0.75pt,yscale=-1,xscale=1]
	
	\draw   (41,29) -- (570.33,29) -- (570.33,357.8) -- (41,357.8) -- cycle ;
	\draw   (41,49.25) -- (95.5,49.25) -- (95.5,137.75) -- (41,137.75) -- cycle ;
	\draw  [dash pattern={on 4.5pt off 4.5pt}] (41,33) -- (68.5,33) -- (68.5,49.25) -- (41,49.25) -- cycle ;
	\draw  [dash pattern={on 4.5pt off 4.5pt}] (41,137.75) -- (113,137.75) -- (113,154.25) -- (41,154.25) -- cycle ;
	\draw   (41,198.75) -- (158,198.75) -- (158,254.25) -- (41,254.25) -- cycle ;
	
	\draw  [dash pattern={on 4.5pt off 4.5pt}] (41,181.25) -- (132.5,181.25) -- (132.5,198.75) -- (41,198.75) -- cycle ;
	\draw   (41,339.8) -- (402.5,339.8) -- (402.5,357.8) -- (41,357.8) -- cycle ;
	\draw  [dash pattern={on 4.5pt off 4.5pt}] (41,321) -- (282.5,321) -- (282.5,339.8) -- (41,339.8) -- cycle ;
	\draw   (41,281.16) -- (238,281.16) -- (238,318.6) -- (41,318.6) -- cycle ;
	\draw    (594.41,90.92) -- (607.67,90.63) ;
	\draw [shift={(609.67,90.59)}, rotate = 178.76] [color={rgb, 255:red, 0; green, 0; blue, 0 }  ][line width=0.75]    (10.93,-3.29) .. controls (6.95,-1.4) and (3.31,-0.3) .. (0,0) .. controls (3.31,0.3) and (6.95,1.4) .. (10.93,3.29)   ;
	\draw    (586.96,47) -- (594.5,46.83) ;
	\draw    (586.78,135.17) -- (594.32,135) ;
	\draw    (594.5,46.83) -- (594.32,135) ;
	\draw    (594.31,227.86) -- (608.33,227.65) ;
	\draw [shift={(610.33,227.62)}, rotate = 179.16] [color={rgb, 255:red, 0; green, 0; blue, 0 }  ][line width=0.75]    (10.93,-3.29) .. controls (6.95,-1.4) and (3.31,-0.3) .. (0,0) .. controls (3.31,0.3) and (6.95,1.4) .. (10.93,3.29)   ;
	\draw    (586.49,196.62) -- (594.41,196.5) ;
	\draw    (586.31,259.33) -- (594.22,259.22) ;
	\draw    (594.41,196.5) -- (594.22,259.22) ;
	\draw    (596.34,305.28) -- (608,305.14) ;
	\draw [shift={(610,305.11)}, rotate = 179.31] [color={rgb, 255:red, 0; green, 0; blue, 0 }  ][line width=0.75]    (10.93,-3.29) .. controls (6.95,-1.4) and (3.31,-0.3) .. (0,0) .. controls (3.31,0.3) and (6.95,1.4) .. (10.93,3.29)   ;
	\draw    (589.67,284.56) -- (596.42,284.5) ;
	\draw    (589.51,323.14) -- (596.26,323.05) ;
	\draw    (596.42,284.5) -- (596.26,323.05) ;
	\draw    (594.44,352.65) -- (608,352.6) ;
	\draw [shift={(610,352.59)}, rotate = 179.78] [color={rgb, 255:red, 0; green, 0; blue, 0 }  ][line width=0.75]    (10.93,-3.29) .. controls (6.95,-1.4) and (3.31,-0.3) .. (0,0) .. controls (3.31,0.3) and (6.95,1.4) .. (10.93,3.29)   ;
	\draw    (586.85,344.86) -- (594.54,344.83) ;
	\draw    (586.67,360.5) -- (594.35,360.47) ;
	\draw    (594.54,344.83) -- (594.35,360.47) ;
	
	\draw (48,35) node [anchor=north west][inner sep=0.75pt]   [align=left] {$\displaystyle x_{4}$};
	\draw (48,85) node [anchor=north west][inner sep=0.75pt]   [align=left] {$\displaystyle x_{5} \ta$};
	\draw (98.5,48) node [anchor=north west][inner sep=0.75pt]   [align=left] {$\displaystyle \ta$};
	\draw (98.5,57.5) node [anchor=north west][inner sep=0.75pt]   [align=left] {$\displaystyle \tb$};
	\draw (97.5,75) node [anchor=north west][inner sep=0.75pt]   [align=left] {$\displaystyle \vdots $};
	\draw (98.5,122) node [anchor=north west][inner sep=0.75pt]   [align=left] {$\displaystyle \tb$};
	\draw (71,35) node [anchor=north west][inner sep=0.75pt]   [align=left] {$\displaystyle \ta\ta$};
	\draw (42.5,152) node [anchor=north west][inner sep=0.75pt]   [align=left] {$\displaystyle \vdots $};
	\draw (48,140) node [anchor=north west][inner sep=0.75pt]   [align=left] {$\displaystyle x_{6} \ta$};
	\draw (115,140) node [anchor=north west][inner sep=0.75pt]   [align=left] {$\displaystyle \ta$};
	\draw (48,215) node [anchor=north west][inner sep=0.75pt]   [align=left] {$\displaystyle x_{2( k-h) +1} \ta$};
	\draw (160,198) node [anchor=north west][inner sep=0.75pt]   [align=left] {$\displaystyle \ta$};
	\draw (160,208) node [anchor=north west][inner sep=0.75pt]   [align=left] {$\displaystyle \tb$};
	\draw (161,216) node [anchor=north west][inner sep=0.75pt]   [align=left] {$\displaystyle \vdots $};
	\draw (160,243) node [anchor=north west][inner sep=0.75pt]   [align=left] {$\displaystyle \tb$};
	\draw (43,254) node [anchor=north west][inner sep=0.75pt]   [align=left] {$\displaystyle \vdots $};
	\draw (48,182) node [anchor=north west][inner sep=0.75pt]   [align=left] {$\displaystyle x_{2( k-h)} \ta$};
	\draw (135,182) node [anchor=north west][inner sep=0.75pt]   [align=left] {$\displaystyle \ta$};
	\draw (48,342) node [anchor=north west][inner sep=0.75pt]   [align=left] {$\displaystyle x_{2k-1} \ta$};
	\draw (404.9,345) node [anchor=north west][inner sep=0.75pt]   [align=left] {$\displaystyle \ta$};
	\draw (48,323) node [anchor=north west][inner sep=0.75pt]   [align=left] {$\displaystyle x_{2( k-1)} \ta$};
	\draw (285.8,325) node [anchor=north west][inner sep=0.75pt]   [align=left] {$\displaystyle \ta$};
	\draw (48,289.66) node [anchor=north west][inner sep=0.75pt]   [align=left] {$\displaystyle x_{2( k-1) -1} \ta$};
	\draw (240,280) node [anchor=north west][inner sep=0.75pt]   [align=left] {$\displaystyle \ta$};
	\draw (240,305) node [anchor=north west][inner sep=0.75pt]   [align=left] {$\displaystyle \tb$};
	\draw (43,350) node [anchor=north west][inner sep=0.75pt]   [align=left] {$\displaystyle \vdots $};
	\draw (43,5.5) node [anchor=north west][inner sep=0.75pt]   [align=left] {$\displaystyle \vdots $};
	\draw (573,343) node [anchor=north west][inner sep=0.75pt]   [align=left] {$\displaystyle \tb$};
	\draw (573,325) node [anchor=north west][inner sep=0.75pt]   [align=left] {$\displaystyle \ta$};
	\draw (573,305) node [anchor=north west][inner sep=0.75pt]   [align=left] {$\displaystyle \tb$};
	\draw (573,280) node [anchor=north west][inner sep=0.75pt]   [align=left] {$\displaystyle \tb$};
	\draw (240,293) node [anchor=north west][inner sep=0.75pt]   [align=left] {$\displaystyle \tb$};
	\draw (573,293) node [anchor=north west][inner sep=0.75pt]   [align=left] {$\displaystyle \tb$};
	\draw (573,254) node [anchor=north west][inner sep=0.75pt]   [align=left] {$\displaystyle \vdots $};
	\draw (573,180) node [anchor=north west][inner sep=0.75pt]   [align=left] {$\displaystyle \ta$};
	\draw (573,240) node [anchor=north west][inner sep=0.75pt]   [align=left] {$\displaystyle \tb$};
	\draw (573,195) node [anchor=north west][inner sep=0.75pt]   [align=left] {$\displaystyle \tb$};
	\draw (573,208) node [anchor=north west][inner sep=0.75pt]   [align=left] {$\displaystyle \vdots $};
	\draw (573,152) node [anchor=north west][inner sep=0.75pt]   [align=left] {$\displaystyle \vdots $};
	\draw (573,140) node [anchor=north west][inner sep=0.75pt]   [align=left] {$\displaystyle \ta$};
	\draw (573,122) node [anchor=north west][inner sep=0.75pt]   [align=left] {$\displaystyle \tb$};
	\draw (573,48) node [anchor=north west][inner sep=0.75pt]   [align=left] {$\displaystyle \tb$};
	\draw (573,75) node [anchor=north west][inner sep=0.75pt]   [align=left] {$\displaystyle \vdots $};
	\draw (573,35) node [anchor=north west][inner sep=0.75pt]   [align=left] {$\displaystyle \ta$};
	\draw (630,75) node [anchor=north west][inner sep=0.75pt]   [align=left,rotate=270] {$\displaystyle F_{2k-5}$};
	\draw (630,195) node [anchor=north west][inner sep=0.75pt]   [align=left,rotate=270] {$\displaystyle F_{2(k-h)+1}$};
	\draw (630,296.33) node [anchor=north west][inner sep=0.75pt]   [align=left,rotate=270] {$\displaystyle F_{3}$};
	\draw (630,342.67) node [anchor=north west][inner sep=0.75pt]   [align=left,rotate=270] {$\displaystyle F_{0}$};

\end{tikzpicture}
}	\caption{The middle part $BWT_{mid}(\widehat{s})$ of the BWT matrix for the deletion of the last character of a Fibonacci word of even order $2k$ is shown. 
	}\label{fig:fib_deletion}
\end{figure}

Fig. \ref{fig:fib_deletion} displays the structure of the middle part of the BWT of $\widehat{s}$. 
To determine the number of runs in this middle part of the matrix, it is crucial first to consider that all the rotations starting with $x_h$, with $h=4,\ldots,2k-1$, are grouped together in the BWT matrix. Specifically, these occur in blocks where rotations starting with $x_{h}\a$, $h$ odd, are immediately preceded by the unique rotation starting with $x_{h-1}\a\a$, and immediately followed by the rotation starting with $x_{h+1}\a\a$, as illustrated in Fig. \ref{fig:fib_deletion}. This is proved in the following lemma.

\bigskip 
\begin{lemma}\label{lemma:mp_allxa-prefixed rotations}
	For $4\leq h \leq 2k-1$, rotations starting with some $x_{h}\a\a$ are smaller then rotations starting with $x_{h+1}\a$.
	In particular, the word $\widehat{s} = x_{2k}\a$ is greater than any of the aforementioned rotations. Moreover, if $h$ is odd, rotations starting with some $x_{h}\a$ are smaller then rotations starting with $x_{h+1}$.
\end{lemma}

\begin{proof}
Every $x_{h}$ is a prefix of $x_{h+1}$.
Since there is exactly one circular occurrence of $\a\a\a$ in $\widehat{s}$, then $x_{h+1}\a$ is either prefixed by $x_{h}\a\b$ or by $x_{h}\b\a$, i.e.\ the $\a\a\a$ factor occurs earlier in $x_{h}\a\a$. In both cases, the first claim holds. 
Finally, the thesis follows by observing that if $h$ is odd, $x_h\b\a$ is a prefix of $x_{h+1}$. 
\end{proof}

\bigskip

\begin{lemma}\label{lemma:mp_oddxs}
There are $F_{2h}-1$ occurrences of $\b x_{2(k-h)}\a$ and $F_{2h+1}$ occurrences of $\b x_{2(k-h)-1}\a$ as circular factors in $\widehat{s} = x_{2k}\a$. 
\end{lemma}

\begin{proof}
    The claim can be proved by induction.     
    For $h=0$, the statement follows from the fact that there is one occurrence of $\b x_{2k}\a$ in the Fibonacci word of order $2k$, therefore there are $F_0 -1 = 0$ occurrences in $\widehat{s}$ because of the missing $\b$ at the end. There are $F_1 = 1$ occurrences of $\b x_{2k-1}\a$ in both words.
    Note that the occurrence of any $\b x_{2h}\a$ in position $F_{2k}-1$ of the Fibonacci word of order $2k$ is missing in $\widehat{s}$ due to the missing $\b$ at the end of the word.
    
    Let us suppose the statement holds for all $i \leq h$.
    The factor $\b x_{2(k-h)-2}\a$ appears as a prefix of $\b x_{2(k-h)-1}\a$ and as a prefix of $\b x_{2(k-h)}\a$. 
    Moreover, the mentioned occurrences are distinct because $\b x_{2(k-h)-1}\a$ is not a prefix of $\b x_{2(k-h)}\a$.
    Therefore, by induction, the number of occurrences of $\b x_{2(k-h)}\a$ is equal to the sum of the number of occurrences of $\b x_{2(k-h)-2}\a$ and those of $\b x_{2(k-h)-1}\a$: $F_{2h}-1+F_{2h+1}= F_{2h+2}-1$.
    On the other, $\b x_{2(k-h)-3}\a$ appears as a suffix of $\b x_{2(k-h)-2}\a$ and as a suffix of $\b x_{2(k-h)-1}\a$. 
    Moreover, the mentioned occurrences are distinct because $\b x_{2(k-h)-2}\a$ is not a suffix of $\b x_{2(k-h)-1}\a$.
    Finally, $\b x_{2(k-h)-3}\a$ appears once also as suffix of $\a x_{2(k-h)-2} \a$, starting in position $F_{2k}-2$ of $\widehat{s}$.
    Therefore, by induction, the number of occurrences of $\b x_{2(k-h)-3}\a$ is equal to the sum of the number of occurrences of $\b x_{2(k-h)-2}\a$ and those of $\b x_{2(k-h)-1}\a$ plus one: $F_{2h+2}-1+F_{2h+1}+1= F_{2h+3}$.
\end{proof}

\bigskip
\begin{lemma}[Middle part]\label{lemma:middle}
	The middle part contributes to $r(v)$ with $2(k-2)$ runs in the following form: $\a\b^{F_{2k-5}} \a\b^{F_{2k-7}} \cdots \a\b^{F_3} \a\b$. 
\end{lemma}

\begin{proof}
By construction, all the rotations starting with $\a\b$ are prefixed by a central word $x_h$, for some $h<2k$. 
For all $\conj_i(\widehat{s})$ such that $\lcp(\widehat{s}, \conj_i(\widehat{s})) = x_h$ for some odd $h$, $\conj_i(\widehat{s})$ is prefixed by $x_h\a$ since $\widehat{s}$ is prefixed by $x_h\b$, then $\conj_i(\widehat{s})$ is in the middle part (i.e.\ smaller than $\widehat{s}$ by Lemma \ref{lemma:mp_allxa-prefixed rotations}) since the word $\widehat{s}$ separates the middle and the bottom part. For all $\conj_i(\widehat{s})$ such that $\lcp(\widehat{s}, \conj_i(\widehat{s})) = x_h$ for some even $h$,  $\conj_i(\widehat{s})$ is prefixed by $x_h\b$ since $\widehat{s}$ is prefixed by $x_h\a$, then  $\conj_i(\widehat{s})$ is in the bottom part (i.e.\ greater than $\widehat{s}$ by Lemma \ref{lemma:mp_allxa-prefixed rotations})
By using Lemma \ref{le:structure_s}, it follows that the shortest non-empty central word starting with $\a\b$  that appears in $\widehat{s}$ as a circular factor is $x_4$. One can prove that, among rotations starting with the same $x_h$, $h \geq 4$ even, the smallest one is preceded by $\a$. In fact, it starts with $x_4 \a\a$. All the following $F_{2k-h-1}$ rotations starting with $x_{h+1}\a$ are preceded by $\b$ (Lemma~\ref{lemma:mp_oddxs}).
The fact that there exist exactly $k-2$ such $x_h$ proves the claim. 
\end{proof}

The rotations that divide the middle part from the bottom part are the two rotations prefixed by the two occurrences of $x_{2k-1}$. 
By properties of Fibonacci words, one rotation is prefixed by $x_{2k-1}\a$ (end middle part) and the other by $x_{2k-1}\b$ (beginning bottom part).
The latter follows the first in lexicographic order. 
Note that the rotation starting with $x_{2k-1}\b$ is $\widehat{s}$, namely $x_{2k-1}\b\a x_{2k-2}\a$.

\bigskip
\begin{lemma}[Bottom part]\label{lemma:bottom}
	All rotations greater than $v = x_{2k-1}\b\a x_{2k-2}\a$ end with \a.
\end{lemma}
\begin{proof}
	From Lemma~\ref{lemma:top} we have that $k-1+F_{2k-3}-k+1$ rotations ending with $\b$ have already appeared in the matrix, and from Lemma~\ref{lemma:middle} $F_{2k-5}+ \ldots +F_3 + F_1$ rotations ending with $\b$ have already appeared in the matrix.
    Summing the number of $\b$'s we have $k-1+F_{2k-3}-k+1+ F_{2k-5}+ \ldots +F_3 + F_1 = F_{2k-3} + F_{2k-5}+ \ldots +F_3 + F_1 $. We can decompose each odd Fibonacci number $F_{2x+1}$ in the sum $F_{2x} + F_{2x-1}$.
    Therefore, the previous sum becomes $F_{2k-4} + F_{2k-5}+ F_{2k-6} + F_{2k-7} \ldots +F_2 + F_1 + F_1$.
    For every Fibonacci number $F_x$, it holds that $F_x = F_{x-2} + F_{x-3} + F_{x-4} + \ldots + F_2 + F_1 + 2$.
    Therefore, $F_{2k-4} + F_{2k-5}+ F_{2k-6} + F_{2k-7} \ldots +F_2 + F_1 + F_1 = F_{2k-2} - 1$, which is exactly the number of $\b$'s in $\widehat{s}$. 
    Therefore all the remaining rotations end with \a. 
\end{proof}

In the context of repetitiveness measures of words, a measure $\lambda$ is called \emph{monotone} if, for each word $v\in\Sigma^*$ and for each letter $\x\in\Sigma$, it holds that $\lambda(v) \leq \lambda(v\x)$. 
Since we have shown that appending or deleting a single character can substantially increase the parameter $r$, the following known result on the monotonicity of $r$ can be derived:

\bigskip
\begin{corollary}\label{coro:monotone}
    The measure $r$ is not monotone.
\end{corollary}

\subsection{Substituting a character}\label{subsec:fibsubstitution}

In this subsection, we show how to increment $r$ by a logarithmic factor by substituting a character.
Consider a Fibonacci word $s$ of even order in which we replace the last $\b$ by an $\a$. Denoting this word by $s'$, we will prove that $\bwt(s')$ has $\Theta(\log n)$ runs, where $n$ is the length of the word. We start with the following lemma in which we assess how the number of BWT-runs changes when we append or prepend to a Lyndon word a character that is smaller than or equal to the smallest character appearing in the word itself.

\bigskip

\begin{lemma}\label{le:pres_order_prep_smaller_primitive}
Let $v\in\Sigma^*$ be a Lyndon word containing at least two distinct letters and let $\x\in \Sigma$ be smaller than or equal to the smallest character occurring in $v$.
Then, $r(v) \leq r(\x v) = r(v \x)\leq r(v)+2$.
\end{lemma}
\bigskip

\begin{proof}
    We can obtain the lexicographic order of the rotations of $\x v$, or equivalently $v \x$, from the order of the rotations of $v$.
    To do so, we show that given two rotations $conj_i(v)<conj_j(v)$, with $i\neq j$, if $conj_i(v)<conj_j(v)$ then $v[i..|v|-1]\x v[0..i-1]<v[j..|v|-1]\x v[0..j-1]$.
    
    Note that $v$ is the smallest rotation in its BWT matrix. Let us denote by $conj_h(v)$, for some $h$, the second rotation in the BWT matrix.
    Since $v$ is primitive, there exist a unique circular factor $u$ smaller than all the other circular factors having the same length.
    In fact, if $t=|\lcp(v,conj_h(v))|$, then $u=v[0..t]$. 
    Moreover, for all $\ell < |u|$, $u[0..\ell-1]$ is the smallest circular factor of length $\ell$ occurring in $v$.
    We can then distinguish two cases.
    
    The first case is when $|\lcp(conj_i(v),conj_j(v))|<\min\{|v|-i+1, |v|-j+1\}$. 
    Under this condition, it follows that the insertion of the $\x$ does not affect the relative order between $v[i..|v|-1]\x v[0..i-1]$ and $v[j..|v|-1]\x v[0..j-1]$.

    Otherwise, if $|\lcp(conj_i(v),conj_j(v))|\geq\min\{|v|-i+1, |v|-j+1\}$, note that $i>j$, i.e. $|v[i..|v|-1]|<|v[j..|v|-1]|$.
    This follows by observing that both $v[i..|v|-1]$ and $v[j..|v|-1]$ are (circularly) followed by $u$ that is unique, and by contradiction if $i<j$ then $u$ would circularly occur before in $conj_j(v)$ with respect to $conj_i(v)$, which contradicts $conj_i(v) < conj_j(v)$. 
    We can now further distinguish between two subcases: when either (i) $u$ is a prefix of $v[0..i-1]$ or (ii) $v[0..i-1]$ is a proper prefix of $u$.
    
    For the subcase (i), as $|\lcp(conj_i(v), conj_j(v))| \ge |v[i..|v|-1]|$, and the factor $u$ is a prefix of $v[0..i-1]$, the first distinct character between $conj_i(v)$ and $conj_j(v)$ lies within the unique occurrence of $u$ in $conj_i(v)$. 
    After the letter $\x$ is inserted, $conj_i(v)$ becomes $v[i..|v|-1]\x v[0..i-1]$, yielding a factor $\x u$ occurring at position $|v|-i+1$ that is also unique and smallest among all the factors of length $|\x u|$ in $\x v$.Whatever factor appears in $v[j..|v|-1]\x u$ at position $|v|-i+1$, has to be greater than $\x u$, and the order is preserved.
 For the subcase (ii), recall that since $v[0..i-1]$ is a proper prefix of $u$, $v[0..i-1]$ is also the smallest $i$-length circular factor in lexicographical order occurring in $v$, but differently from $u$ the circular factor $v[0..i-1]$ is repeated (otherwise $|u| \leq |v[0..i-1]|$, contradiction).
 By primitivity of $v$, the first distinct character between $conj_i(v)$ and $conj_j(v)$ lies within $v[0..i-1]$, i.e., within $conj_i(v)[|v|-i+1..|v|-1]$.
 After the insertion of the symbol $\x$ the analogous behavior of subcase (i) is observed.

We conclude the proof by observing that, with respect to the original BWT, we have one extra rotation, and one rotation for which the letter in the BWT has changed, which are $\x v$ and $v \x$ respectively.
Observe that by construction, $\x v$ is now the smallest among all of its rotation, which ends with the last letter of $v$.
On the other hand, $v \x$ is now the second smallest rotation and it ends with $\x$.
Hence, $\bwt(\x v) = \bwt(v)[0]\cdot \x \cdot \bwt(v)[1..|v|-1]$, and the thesis follows. 
\end{proof}

\bigskip
\begin{proposition}\label{prop:subst_infib}
Let $s$ be the Fibonacci word of even order $2k>4$, and $n=|s|$. Let $s'$ be the word resulting from substituting a $\b$ by an $\a$ at position $F_{2k}-1$. Then $\bwt(s')$ has $2k+2$ runs. 
\end{proposition}
\bigskip

\begin{proof} Observe that $s' = \widehat{s}\a$. By Proposition \ref{prop:removedC_bwtForm}, it holds $r(\widehat{s}) = 2k$. By Lemma \ref{lemma:top}, we know that $conj_{n-3}(\widehat{s})$ is the smallest rotation of $\widehat{s}$. By Lemma \ref{le:structure_s},  $conj_{n-3}(\widehat{s})=\a\a x_{2k-1}\cdots \widehat{x_4}$. 
By Lemma \ref{le:pres_order_prep_smaller_primitive}, it holds $2k\leq r(\a conj_{n-3}(\widehat{s})) \leq 2k+2$. More precisely, it is $2k+2$ since $\a conj_{n-3}(\widehat{s})$ is the smallest rotation of its BWT matrix, $conj_{n-3}(\widehat{s})\a$is the second smallest rotation, and the relative order among the rotations of $\a conj_{n-3}(\widehat{s})$ coincide with that of the rotations of $\widehat{s}$, using the same argument from the proof of Lemma \ref{le:pres_order_prep_smaller_primitive}. This means that to obtain $BWT(s')$, it is enough to insert an $\a$ between the first two $\b$'s in $BWT(\widehat{s})$. As $r(\a conj_{n-3}(\widehat{s}))=r(\a\widehat{s}) = r(\widehat{s}\a)$, the thesis follows.
\end{proof}

\section{Additive $\Theta(\sqrt{n})$ factor}
\label{sec:additive_sqrt_n}

In the previous section we proved that a single edit operation can cause a multiplicative increase by a logarithmic factor in the number of runs. In this section, we will exhibit an infinite family of words on which a single edit operation can cause an additive increment of $r$ by $\Theta(\sqrt{n})$ (see Def.~\ref{def:w_k} below).

As we saw in the previous section, there exist infinite families of words such that $r=\Theta(\log n)$, where $n$ is the length of the word. Other families with logarithmic number of runs of the BWT are also known from the literature, e.g.\ the Thue-Morse words \cite{BrlekFMPR19}. 
Moreover, there exist words such that $r$ is maximal, i.e., $r(w)=|w|$. For instance, if $w=\a\a\a\a\b\b\a\b\a\b\b\b\b\a\a\b\a\b$, then $r(w)=18=|w|$ \cite{MANTACI_TCS2017}. Next we show that there is no gap between these two scenarios, i.e., it is possible to construct infinite families of words $w$ such that $r(w)=\Theta(n^{1/k})$, for any $k>1$.

\bigskip
\begin{proposition}
Let $k$ be a positive integer. There exists an infinite family $T_k$ of binary words such that $r(w)=\Theta(n^{1/k})$, for any $w\in T_k$. 
\end{proposition}

\begin{proof}
We can define the set $T_k=\{w_{i,k}=\prod_{j=1}^i \a\b^{j^k} \mid i\geq 1\}$. We can state that $|w_{i,k}|=\Theta(i^{k+1})$. Moreover, $r(w_{i,k})=\Theta(i)$
In fact, each $\a$ in $\bwt(w_{i,k})$ corresponds to the last letter of one of the rotations having prefix $\b^{j^k}\a$, for some $1\leq j \leq i$.
On the other hand, all the rotations with prefix $\a\b$, as well as the remaining rotations with prefix $\b^\ell\a$ for all $1\leq \ell \leq i^k$, end with $\b$.
It follows then that whenever $k\geq 2$, all the $\a$'s in $\bwt(w_{i,k})$ are separated by an equal-letter run of $\b's$, leading to $r(w_{i,k})=2i=\Theta(i)$.
However, note that for a fixed $\ell$, the rotations starting with $b^\ell \a$ are sorted according to the length of the maximal run of $\b$'s following the common prefix.
Thus, even for $k=1$, there is only one run of consecutive $\a$'s in $\bwt(w_{i,k})$, while the remaining are separated.
More in detail, $\bwt(w_{i,1})=\b\b^i\a\b^{i-1}\a\cdots \a \b^3\a\b^2\a\a$. Hence, $r(w_{i,1})=2i-2=\Theta(i)$.
The claim follows by observing that for the family of words $w_{i,k}$ it holds that $r(w_{i,k})=\Theta(i)=\Theta(n^{\frac{1}{k+1}})$, where $n=|w_{i,k}|$.
 \end{proof}

\bigskip

We will show that the following family of words satisfies that a single edit operation can cause an additive increment of $r$ by $\Theta(\sqrt{n})$.

\bigskip
\begin{definition}\label{def:w_k}
For any $k > 5$,  let $s_i = \a\b^i\a\a$ and $e_i = \a\b^i\a\b\a^{i-2}$ for all $2 \leq i \leq k-1$, and $q_k = \a\b^k\a$.
Then, $$w_k = (\prod_{i=2}^{k-1}s_i e_i)q_k = (\prod_{i=2}^{k-1}\a\b^i\a\a \a\b^i\a\b\a^{i-2})\a\b^k\a.$$
\end{definition}
\bigskip

The length of these words can be easily deduced from their definition.

\bigskip
\begin{observation}\label{le:w_k_length}
Let $n = |w_k|$ for some $k > 5$. It holds that $ n = \sum_{i = 2}^{k-1}(3i+4) + (k+2) = (3k^2+7k-18)/2$. Moreover, it holds that $k = \Theta(\sqrt{n})$.
\end{observation}
\bigskip

The following lemma will be used to show how the rotations of $w_k$ can be sorted according to the factorization $s_2 e_2 \cdots s_{k-1} e_{k-1} q_k$.

\bigskip
\begin{lemma}
\label{le:sorting_s_e_q}
    Let $k>5$ be an integer.
    Then, $s_2 < e_2 < s_3 < e_3 <\ldots < s_{k-1} < e_{k-1} < q_k$.
    Moreover the set $\mathcal{U} = \bigcup_{i = 2}^{k-1}\{s_i, e_i\} \cup \{q_k\}$ is prefix-free.
\end{lemma}
\begin{proof}For the first claim, note from the definition of the words $e_i, s_i$ and $q_k$ that for $i \in [2, k-1]$ it holds $s_i < e_i$, for $i \in [2, k-2]$ it holds $e_i < s_{i+1}$, and it holds $e_{k-1} < q_k$. For the second claim, observe that for any two distinct strings $x$ and $y$ in the set $\mathcal{U}$ starting with $\a\b^j\a$ and $\a\b^{j'}\a$ respectively, there are two possible cases. If $j = j'$ then $x$ and $y$ are $s_i$ and $e_i$ respectively, and none of them is a prefix of the other. Otherwise, w.l.o.g. $j < j'$, so $x = \a\b^j\a x'$ and $y = \a\b^{j}\b y'$ for some $x'$ and $y'$. Hence $x[j+2] \neq y[j+2]$ and none of them is a prefix of the other. Thus, the set $\mathcal{U}$ is prefix-free.
\end{proof}

\subsection{Characterizing the BWT of $w_k$}

In order to characterize the BWT of the word $$w_k = (\prod_{i=2}^{k-1}s_i e_i)q_k = (\prod_{i=2}^{k-1} \a\b^i\a\a \cdot \a\b^i\a\b\a^{i-2})\cdot \a\b^k\a,$$ we divide its BWT matrix into disjoint ranges of consecutive rotations sharing the same (specific) prefixes, and characterize the substring of $\bwt(w_k)$ corresponding to each one of these prefixes.

\bigskip
\begin{definition}
Given $x,w\in \Sigma^*$, we denote by $\beta(x,w)$ the substring of $\bwt(w)$ corresponding to the range of contiguous rotations prefixed by $x$. We omit the second parameter of  $\beta(x,w)$ when it is clear from the context. 
\end{definition}\bigskip

The structure of the whole BWT matrix of $w_k$ is summarized in Table \ref{table:wk_bwt_structure}. The following series of lemmas characterize the substring of $\bwt(w_k)$ corresponding to each range to be considered.

\bigskip
\begin{lemma}[$\vv{\a^{k-2}\b}{}$]
\label{le:a^k-2}
    Given the word $w_k= (\prod_{i=2}^{k-1}s_i e_i)q_k$ for some $k > 5$, the first rotation in the BWT matrix is $\a^{k-3}q_k \cdots \b$.
\end{lemma}
\begin{proof}
    The first rotation in lexicographic order must start with the longest run of $\a$'s.
    By definition of $w_k$, the longest run of $\a$'s has length $k-2$, and it is obtained by concatenating the suffix $\a^{k-3}$ of $e_{k-1}$ with $q_k$, which is preceded by a $\b$ (otherwise we could extend the run of $\a$'s). 
\end{proof}

\begin{lemma}[$\vv{\a^i\b}{}$ for $4\leq i \leq k-3$]
    \label{le:a^ib}
    Given the word $w_k= (\prod_{i=2}^{k-1}s_i e_i)q_k$ for some $k> 5$, and an integer $4 \leq i \leq k-3$,  the rotations in the BWT matrix starting with $\a^i \b$ are $\a^{i-1} s_{i+2} \cdots \b < \a^{i-1} s_{i+3} \cdots \a < \ldots < \a^{i-1} s_{k-1} \cdots \a < \a^{i-1} q_k \cdots \a$.
\end{lemma}
\begin{proof}
    One can notice that, for all $4\leq i\leq k-3$, the (circular) factor $\a^i\b$ can only be obtained, for all $i+2 \leq j \leq k$, from the concatenation of the suffix $\a^{i-1}$ of $e_{j-1}$, with either the prefix $\a\b$ of $s_{j}$, if $i+2 \leq j \leq k-1$, or the prefix $\a\b$ of $q_k$, if $j= k$.
    By Lemma \ref{le:sorting_s_e_q}, we can sort these rotations according to the lexicographic order of $\bigcup_{j = i}^{k-1}\{s_j\} \cup \{q_k\}$.
    Note that all these rotations end with an $\a$, with the exception of the rotation starting with $\a^{i-1} s_{i+2}$, since it is where the only occurrence of $\b \a^i \b$ can be found.  
\end{proof}

\begin{lemma}[$\vv{\a\a\a\b}{}$]
\label{le:a^3b}
    Given the word $w_k= (\prod_{i=2}^{k-1}s_i e_i)q_k$ for some $k > 5$, the first five rotations in the BWT matrix starting with $\a\a\a \b$ are $\a\a e_2 \cdots \b < \a\a e_3 \cdots \b < \a\a e_4 \cdots \b < \a\a s_5 \cdots \b < \a\a e_5 \cdots \b$, while the remaining are $\a\a s_{6} \cdots \a < \a\a e_{6} \cdots \b < \ldots < \a\a s_{k-1} \cdots \a < \a\a e_{k-1} \cdots \b < \a\a q_k \cdots \a$.
\end{lemma}
\begin{proof}
    Analogously to the proof of Lemma \ref{le:a^ib}, some of the rotations starting with $\a\a\a \b$ can be obtained, for all $5 \leq j \leq k$, from the concatenation of the suffix $\a\a$ of $e_{j-1}$, with either the prefix $\a\b$ of $s_j$, if $5 \leq j \leq k-1$, or the prefix $\a\b$ of $q_k$, if $j = k$.
    However, in this case we have more rotations starting with $\a \a \a \b$, that are those rotations starting with the suffix $\a \a$ of $s_{j'}$ concatenated with the prefix $\a\b$ of $e_{j'}$, for all $2\leq j'\leq k-1$.
    Thus, all the rotations starting with $\a\a\a\b$ are sorted according to the lexicographic order of the words in $\bigcup_{j = 5}^{k-1}\{s_j\} \cup \bigcup_{j'=2}^{k-1}\{e_{j'}\} \cup \{q_k\}$.
    Note that all the rotations starting either with $\a \a s_j$, for all $6 \leq j \leq k-1$, or with $\a \a q_k$, end with $\a$.
    On the other hand, the rotations starting either with $\a\a s_5$ or with $\a \a e_j$, for all $2\leq j \leq k-1$, end with a $\b$.   
\end{proof}

\begin{lemma}[$\vv{\a\a\b}{}$]
\label{le:a^2b}
    Given the word $w_k= (\prod_{i=2}^{k-1}s_i e_i)q_k$ for some $k > 5$, the first five rotations in the BWT matrix starting with $\a\a \b$ are $\a s_2 \cdots \b < \a e_2 \cdots \a < \a e_3 \cdots \a < \a s_4 \cdots \b < \a e_4 \cdots \a$, while the remaining are $\a s_{5} \cdots \a < \a e_{5} \cdots \a < \ldots < \a s_{k-1} \cdots \a < \a e_{k-1} \cdots \a < \a q_k \cdots \a$.
\end{lemma}
\begin{proof}
    Each of the rotations starting with $\a\a \a\b$ from Lemma \ref{le:a^3b} induces a rotation starting with $\a\a\b$, obtained by shifting on the left one character $\a$.
    It follows that all of these rotations end with an $\a$.
    The other rotations starting with $\a \a \b$ are the one obtained by concatenating the suffix $\a$ of $e_3$ and the prefix $\a\b$ of $s_4$, and the one obtained by concatenating the suffix $\a$ of $q_k$ and the prefix $\a\b$ of $s_2$.
    Moreover, both the rotations end with a $\b$.
    The thesis follows by sorting the rotations according to the lexicographic order of the words in $\{s_2\} \cup \bigcup_{j = 4}^{k-1}\{s_j\} \cup \bigcup_{j'=2}^{k-1}\{e_{j'}\} \cup \{q_k\}$.
\end{proof}

\begin{lemma}[$\vv{\a\b}{}$]
\label{le:ab}
    Given the word $w_k= (\prod_{i=2}^{k-1}s_i e_i)q_k$ for some $k > 5$, the first $k-2$ rotations in the BWT matrix starting with $\a \b$ are $\a\b\a^{k-3} q_k \cdots \b < \a\b\a^{k-4} s_{k-1} \cdots \b < \ldots < \a\b s_3 \cdots \b$, the following four rotations are $s_{2} \cdots \a < e_{2} \cdots \a < s_3 \cdots \b < e_{3} \cdots \a$, and the remaining are $s_4\cdots \a < e_{4} \cdots \a < \ldots < s_{k-1} \cdots \a< e_{k-1}\cdots \a < q_k \cdots \a$.
\end{lemma}
\begin{proof}
    For any two distinct integers $i, i'\geq 0$, we have that $\a\b\a^{i}\b < \a\b\a^{i'}\b$ if and only if $i > i'$.
    Thus, the first rotation in lexicographic order starting with $\a\b$ is the one which is followed by the longest run of $\a$'s.
    The smallest of these rotations can be found by concatenating the suffix $\a\b\a^{k-3}$ of $e_{k-1}$ with the prefix $\a\b$ of $q_k$, followed by the suffix $\a\b\a^{i-2}$ of $e_{i-1}$ concatenated with the prefix $\a\b$ of $s_{i}$, for all $3\leq i\leq k-1$ taken in decreasing order.
    By construction of $e_i$, for all $3\leq i\leq k-1$, these rotations must end with a $\b$.
    
    The remaining rotations starting with $\a\b$ are exactly those rotations having as prefix either $s_i$ or $e_i$, for all $2 \leq i \leq k-1$, or $q_k$.
    Note that all of these rotations are obtained by shifting on the left one character $\a$ from the rotations starting with $\a\a \b$ from Lemma \ref{le:a^2b}, with the exception of the one starting with $s_3$.
    It follows that the latter ends with a $\b$, while all the other rotations with an $\a$.
\end{proof}

\begin{lemma}[$\vv{\b\a}{}$]
\label{le:ba}
    Given the word $w_k= (\prod_{i=2}^{k-1}s_i e_i)q_k$ for some $k > 5$, the first $k-5$ rotations in the BWT matrix starting with $\b \a$ are $\b \a^{k-3} q_k \cdots \a < \b \a^{k-4} s_{k-1} \cdots \a  < \ldots < \b \a^3 s_6 \cdots \a$, followed by $\b \a\a e_2 \cdots \b< \b \a\a e_3 \cdots \b< \b \a\a e_4 \cdots \b< \b \a\a s_5 \cdots \a< \b \a\a e_5 \cdots \b$, then by $\b \a\a e_6 \cdots \b < \b\a\a e_7 \cdots \b < \ldots < \b \a\a e_{k-1} \cdots \b < \b \a s_2 \cdots \b < \b \a s_4 \cdots \a$, and finally by $\b \a \b \a^{k-3} q_k \cdots \b < \b \a \b \a^{k-4} s_{k-1} \cdots \b < \ldots <\b \a \b s_{3} \cdots \b < \b s_3 \cdots \a$. 
\end{lemma}
\begin{proof}
    One can notice that we have as many circular occurrences of $\b\a$ as the number of maximal (circular)   runs of $\b$'s in $w_k$.
    Then, for all $2 \leq i \leq k-1$, we have (i) one run of $\b$'s in $s_i$, and (ii) two runs in $e_i$, and (iii) one run in $q_k$.
    
    For the case (i), we have one rotation starting with $\b \a \a e_i$, for each $2\leq i \leq k-1$.
    Since each run of $\b$'s within each word from $\bigcup_{i=2}^{k-1}\{s_i\}$ is of length at least $2$, all rotations in (i) end with a $\b$.
    
    For the case (ii), for all $2 \leq i \leq k-1$, we can distinguish between two sub-cases, based on where $\b \a$ starts: if either (ii.a) from the first run of $\b$'s in $e_i$, or (ii.b) from the second one.
    For the case (ii.a), we can see that these rotations are of the type $\b\a\b\a^{i-2} s_{i+1}$, if $2 \leq i < k-2$, and $\b\a\b\a^{k-3} q_k$.
    Analogously to the case (i), each rotations for case (ii.a) end with a $\b$.
    Each rotation in (ii.b) is obtained by shifting two characters on the right each rotation in (ii.a).
    Therefore, all of these rotations end with an $\a$ and have prefixes of the type $\b\a^{i-2} s_{i+1}$, if $2 \leq i < k-2$, or $\b\a^{k-3} q_k$.

    For the case (iii), the rotation starting with $\b\a$ in $q_k$ has $\b\a s_2$ as prefix, and it is preceded by a $\b$.

    Observe that only for (ii.b) we have rotations starting with $\b \a \a \a \a$.
    Hence, the first rotation in lexicographic order is the one starting with $\b\a^{k-3}q_k$, followed by those starting with $\b\a^{k-4}s_{k-1} < \b\a^{k-5}s_{k-2} < \ldots < \b\a\a\a s_6$.
    
    Among the remaining rotations, those having prefix $\b\a\a\a$ either start with $\b \a \a s_5$ from (ii.b), or $\b \a \a e_i$ from (i), for all $2 \leq i \leq k-1$.
    Thus, by Lemma \ref{le:sorting_s_e_q}, we can sort them according to the order of the words in $\{s_5\} \cup \bigcup_{i=2}^{k-1}\{e_i\}$.
    Then, the remaining rotations with prefix $\b \a \a$ are those starting with $\b \a s_2$ from (iii), and $\b\a s_4$ from (ii.b).
    Finally, let us focus on the rotations from case (ii.a).
    These rotations are sorted according to the length of the run of $\a$'s following the common prefix $\b\a\b$, similarly to the sorting of the rotations from the case (ii.b). 
    The last rotation left is the one starting with $\b s_3$ from case (ii.b).
    Since this rotation is greater than each word from case (ii.a), this is the greatest rotation of $w_k$ starting with $\b\a$ and the thesis follows.
\end{proof}

\begin{lemma}[$\vv{\b^{j}\a}{}$ for all $2 \leq j \leq k-1$]
\label{le:b^ia}
    Given the word $w_k= (\prod_{i=2}^{k-1}s_i e_i)q_k$ for some $k > 5$, and an integer $2 \leq i \leq k-2$, the first $k-i$ rotations in the BWT matrix starting with $\b^i \a$ are $\b^i \a\a e_{i} \cdots \a < \b^i\a\a e_{i+1} \cdots \b < \ldots < \b^{i}\a\a e_{k-1} \cdots \b < \b^i \a s_2 \cdots \b$, followed by $\b^i \a \b \a^{k-3} q_k \cdots \b < \b^i \a \b \a^{k-4} s_{k-1} \cdots \b < \ldots <\b^i \a \b \a^{i-1} s_{i+2} \cdots \b < \b^i \a \b \a^{i-2} s_{i+1} \cdots \a$.
\end{lemma}
\begin{proof}
    All runs of $\b$'s of length at least $2 \leq i \leq k-2$, either appear in (i) $s_j$ or (ii) $e_j$, for all $i \leq j \leq k-1$, or in (iii) $q_k$.
    Let us consider the three cases separately.
    For all $i \leq j \leq k-1$, the rotation starting within $s_j$ (i) has as prefix $\b^i\a\a e_{j}$.
    For all $i \leq j \leq k-2$, the rotation starting within $e_j$ (ii) has as prefix  $\b^i\a\b\a^{j-2} s_{j+1}$, and for $j = k-1$ we have the rotation with prefix $\b^i \a \b\a^{k-3} q_k$.
    Finally, the rotation starting within $q_k$ (iii) has as prefix $\b^i\a s_2$.

    By construction, we can see that first we have all the rotations from case (i) sorted according to the lexicographic order of the words in $\bigcup_{j=i}^{k-1}\{e_i\}$ (Lemma \ref{le:sorting_s_e_q}), then we have the rotation from case (iii), and finally the rotation from case (ii), sorted according to the decreasing length of the run of $\a$'s following the common prefix $\b^i\a\b$.

    Moreover, note that only when the run of $\b$'s is of length exactly $i$ the rotation end with an $\a$.
    Thus, the only for the rotations ending with an $\a$ are those starting within $s_i$ and $e_i$, i.e. those with prefix $\b^i \a e_i$ and $\b^i \a\b\a^{i-2} s_{i+1}$.        
\end{proof}

\begin{lemma}[$\vv{\b^k\a}{}$]
\label{le:b^ka}
    Given the word $w_k= (\prod_{i=2}^{k-1}s_i e_i)q_k$ for some $k > 5$, the last four rotations of the BWT matrix are $\b^{k-1} \a \a e_{k-1} \cdots \a < \b^{k-1} \a s_{2} \cdots \b < \b^{k-1} \a \b \a^{k-3} q_k \cdots \a < \b^{k} \a s_{2} \cdots \a$.
\end{lemma}
\begin{proof}
    Observe that the only rotations with prefix $\b^{k-1}\a$ either start within $s_{k-1}$, or $q_k$, or $e_{k-1}$.
    These rotations have prefix respectively $\b^{k-1}\a\a e_{k-1}$, $\b^{k-1}\a s_2$, and $\b^{k-1}\a\b\a^{k-3}q_k$.
    One can see that these rotations taken in this order are already sorted, and only the rotation starting within $q_k$ ends with a $\b$, while the other two with an $\a$.
    Finally, the only occurrence of $\b^k$ is within $q_k$.
    Hence, the last rotation in lexicographic order starts with $\b^k\a s_2$, and since the run of $\b$'s is maximal it ends with an $\a$, and the thesis follows.
\end{proof}

The following proposition puts together the BWT computations carried out for all blocks of consecutive rows, highlighting which prefixes are shared.

\begin{table}[thb]
\scalebox{.7}{
    \small
    \centering
    \begin{tabular}[t]{|l|l|c|}
        \multicolumn{1}{p{1cm}}{\centering Block prefix} &  \multicolumn{1}{p{1.5cm}}{\centering Ordering factor} & \multicolumn{1}{p{0.5cm}}{\centering BWT} \\
        \hline
        $\a^{k-2}\b$ &  $\b^{k-1}\a$ &  $\b$ \\
        \hline
        \multirow{2}{*}{$\a^{k-3}\b$} & $\b^{k-2}\a\a$ & $\b$\\
        & $\b^{k-1}\a$ & $\a$\\
        \hline
        \multicolumn{1}{|c|}{\vdots} & \multicolumn{1}{c|}{\vdots} &  \vdots  \\ \hline
        \multirow{4}{*}{$\a^4\b$} & $\b^{5}\a\a $  & $\b$ \\
        & $\b^6\a\a$ & $\a$\\
        & $\vdots$ & $\vdots$\\
        & $\b^{k-1}\a$ & $\a$\\
        \hline
        \multirow{11}{*}{$\a\a\a\b$} & $\b\a\b$  & $\b$ \\
        & $\b\b\a\b\a$  & $\b$ \\
        & $\b\b\b\a\b\a\a$ & $\b$ \\
        & $\b\b\b\b\a\a$ & $\b$ \\
        & $\b\b\b\b\a\b\a\a\a$ & $\b$ \\
        \cdashline{2-3}
        & $\b\b\b\b\b\a\a$ & $\a$ \\
        & $\b\b\b\b\b\a\b\a\a\a\a$ & $\b$ \\
        & $\vdots$ & $\vdots$\\
        & $\b^{k-2}\a\a$ & $\a$\\
        & $\b^{k-2}\a\b\a^{k-3}$ & $\b$\\
        & $\b^{k-1}\a$ & $\a$ \\
        \hline
    \end{tabular}
    \quad
        \begin{tabular}[t]{|l|l|c|}
        \multicolumn{1}{p{1cm}}{\centering Block prefix} &  \multicolumn{1}{p{1.5cm}}{\centering Ordering factor} & \multicolumn{1}{p{0.5cm}}{\centering BWT} \\
        \hline
        \multirow{11}{*}{$\a\a\b$} & $\b\a\a$  & $\b$ \\
        & $\b\a\b$  & $\a$ \\
        & $\b\b\a\b\a$  & $\a$ \\
        & $\b\b\b\a\a$  & $\b$ \\
        & $\b\b\b\a\b\a\a$  & $\b$ \\
        \cdashline{2-3}
        & $\b\b\b\b\a\a$  & $\a$ \\
        & $\b\b\b\b\a\b\a\a\a$  & $\a$ \\
        & $\vdots$ & $\vdots$ \\
        & $\b^{k-2}\a\a$  & $\a$ \\
        & $\b^{k-2}\a\b\a^{k-3}$  & $\a$ \\
        & $\b^{k-1}\a$  & $\a$ \\
        \hline
        \multirow{11}{*}{$\a\b$} & $\a^{k-3}q_k$ & $\b$\\
        & $\a^{k-4}s_{k-1}$ & $\b$\\
        & $\vdots$ & $\vdots$ \\
        & $s_3$ & $\b$\\
        \cdashline{2-3}
        & $\b\a\a$ & $\a$\\
        & $\b\a\b$ & $\b$\\
        & $\b\b\a\a$ & $\b$\\
        & $\b\b\a\b\a$ & $\a$\\
        \cdashline{2-3}
        & $\b\b\b\a\a$ & $\a$\\
        & $\b\b\b\a\b\a\a$ & $\a$\\
        & $\vdots$ & $\vdots$ \\
        & $\b^{k-1}\a$ & $\a$\\
        \hline
        \end{tabular}
    \quad
        \begin{tabular}[t]{|l|l|c|}
        \multicolumn{1}{p{1cm}}{\centering Block prefix} &  \multicolumn{1}{p{1.5cm}}{\centering Ordering factor} & \multicolumn{1}{p{0.5cm}}{\centering BWT} \\
        \hline
        \multirow{19}{*}{$\b\a$} & $\a^{k-4} q_k$ & $\a$\\
         & $\a^{k-5} s_{k-1}$ & $\a$\\
         & $\vdots$& $\vdots$ \\
         & $\a^{2} s_6$ & $\a$\\
         \cdashline{2-3}
         & $\a e_2$ & $\b$ \\
         & $\a e_3$ & $\b$ \\
         & $\a e_4$ & $\a$ \\
         & $\a s_5$ & $\a$ \\
         & $\a e_5$ & $\b$ \\
         \cdashline{2-3}
         & $\a e_6$ & $\b$ \\ 
         & $\vdots$ & $\vdots$ \\ 
         & $\a e_{k-1}$ & $\b$ \\ 
         & $s_2$ & $\b$ \\ 
         & $s_4$ & $\a$ \\
         \cdashline{2-3}
         & $\b\a^{k-3}q_k$ & $\b$ \\
         & $\b\a^{k-4}s_{k-1}$ & $\b$ \\
         & $\vdots$ & $\vdots$ \\
         & $\b s_3$ & $\b$ \\
         & $\b\b\b\a\a$ & $\a$ \\
         \hline
        \end{tabular}    
    \quad     
        \begin{tabular}[t]{|l|l|c|}
        \multicolumn{1}{p{1cm}}{\centering Block prefix} &  \multicolumn{1}{p{1.5cm}}{\centering Ordering factor} & \multicolumn{1}{p{0.5cm}}{\centering BWT} \\
        \hline
         \multirow{10}{*}{$\b\b\a$} & $\a e_2$ & $\a$\\
         & $\a e_3$ & $\b$\\
         & $\vdots$ & $\vdots$\\
         & $\a e_{k-1}$ & $\b$\\
         & $s_2$ & $\b$\\
         \cdashline{2-3}
         & $\b\a^{k-3}q_k$ & $\b$\\
         & $\b\a^{k-4}s_{k-1}$ & $\b$\\
         & $\vdots$ & $\vdots$\\
         & $\b\a s_{4}$ & $\b$\\
         & $\b\a s_{3}$ & $\a$\\
         \hline
         \multicolumn{1}{|c|}{\vdots} & \multicolumn{1}{c|}{\vdots} &  \vdots  \\ 
         \hline
         \multirow{3}{*}{$\b^{k-1}\a$} & $\a e_{k-1}$ & $\a$ \\ 
         & $s_2$ &  $\b$ \\ 
         & $\b\a^{k-3} q_k$ &  $\a$ \\
         \hline
         $\b^k\a$ & $s_2$ &  \a \\
         \hline
         \multicolumn{3}{c}{
           }
    \end{tabular}
    }
    \vspace{5pt}
    \caption{Scheme of the BWT-matrix of a word $w_k$ with $k > 5$. The \emph{block prefix} column shows the common prefix shared by all the rotations in a block. The \emph{ordering factor} column shows the factor following the block prefix of a rotation, which decides its relative order inside its block. The BWT column shows the last character of each rotation. The dashed lines divide sub-ranges of rotations for which the BWT follows distinct patterns.}
    \label{table:wk_bwt_structure}
\end{table}

\bigskip
\begin{proposition} \label{prop:bwt_wk}
      Given an integer $k >5$, let $w_k= (\prod_{i=2}^{k-1}s_i e_i)q_k$.
      Then, 
      \begin{align*}
      \vv{\a^i\b}{} &=\b\a^{k-i-2} \text{ for all } 4 \leq i \leq k-2, \\
      \vv{\a^3\b}{} &= \b^5(\a\b)^{k-6}\a, \\
      \vv{\a^{2}\b}{} &= \b\a\a\b\a^{2k-8}, \\
      \vv{\a\b}{} &= \b^{k-2}\a\a\b\a^{2k-6}, \\
      \vv{\b\a}{} &= \a^{k-5}\b\b\b\a\b^{k-4}\a\b^{k-2}\a, \\ 
      \vv{\b^j\a}{} &= \a\b^{2k-2j-1}\a \text{ for all } 2\leq j \leq k-1, \text{ and } \\
      \vv{\b^{k}\a}{} &= \a.
      \end{align*}
      Hence, the BWT of the $w_k$ is $\bwt(w_k) = \prod_{i=2}^{k-1}\vv{\a^{k-i}\b}{}\cdot \prod_{i=1}^{k}\vv{\b^i\a}{}$. Moreover, $r(w_k) = 6k-12$.
\end{proposition}

\begin{proof} 
    The words $\vv{\a^{k-2}\b}{}$, $\vv{\a^i\b}{}$ for all $4 \leq i \leq k-2$, $\vv{\a^{3}\b}{}$, $\vv{\a^{2}\b}{}$, $\vv{\a\b}{}$, $\vv{\b\a}{}$, $\vv{\b^j\a}{}$ for all $2 \leq j \leq k-1$, and $\vv{\b^k\a}{}$, are the concatenations of the last characters of the rotations from Lemma \ref{le:a^k-2}, Lemma \ref{le:a^ib}, Lemma \ref{le:a^3b}, Lemma \ref{le:a^2b}, Lemma \ref{le:ab}, Lemma \ref{le:ba}, Lemma \ref{le:b^ia}, and Lemma \ref{le:b^ka} respectively.
    Moreover, every rotation used to build $\vv{\a^i\b}{}$ is smaller than each rotation used to build $\vv{\a^{i'}\b}{}$, for every $1 \leq i' < i \leq k-2$.
    Symmetrically, every rotation used to build $\vv{\b^j\a}{}$ is greater than each rotation used to build $\vv{\b^{j'}\a}{}$, for every $1 \leq j' < j \leq k$.
    Since we have considered all the disjoint ranges of rotations of $w_k$ based on their common prefix, the word $\prod_{i=2}^{k-1}\vv{\a^{k-i}\b}{}\cdot \prod_{i=1}^{k}\vv{\b^i\a}{}$ is the BWT of $w_k$.

    With the structure of $\bwt(w_k)$, we can easily derive its number of runs.
    The word $\prod_{i = 2}^{k-4}(\vv{\a^{k-i}\b}{})$ has exactly $2(k-6)$ runs: we start with 2 runs from $\vv{\a^{k-2}\b}{}\vv{\a^{k-3}\b}{} = \b\b\a$, and then, concatenating each other $\vv{\a^i\b}{}$ up to  $\vv{\a^4\b}{}$ adds 2 new runs each. It is easy to see that $\vv{\a\a\a\b}{}$, $\vv{\a\a\b}{}$, and $\vv{\a\b}{}$, have $2(k-5)$, $4$, and $4$ runs, respectively. Moreover, the boundaries between these words do not merge, nor with $\vv{\a^4\b}{}$ in the case of $\vv{\a\a\a\b}{}$. The word $\vv{\b\a}{}$ has exactly $7$ runs but it merges with $\vv{\a\b}{}$ and $\vv{\b\b\a}{}$, hence we only charge $5$ runs to this word. The remaining part of the BWT, i.e., $\prod_{i = 2}^{k}(\vv{\b^{i}\a}{})$, has $2(k-2) + 1$ runs: we start with
   3 runs from $\vv{\b\b\a}{}$, and then, concatenating each other $\vv{\b^i\a}{}$ up to  $\vv{\b^{k-1}\a}{}$ adds 2 new runs each. The word $\vv{\b^{k}\a}{}$ does not add new runs, as it consists only of an $\a$ that merges with the previous one. Overall, we have $2(k-6) + 2(k-5) +4 + 4 + 5 + 2(k-2) + 1 = 6k-12$, and the claim holds.    
\end{proof}

\subsection{BWT of $w_k$ after an edit operation}

The following lemmas describe the BWT of $w_k$ after some specific edit operations are applied.
Instead of proving the whole structure of the BWT from the beginning, we compare how the edit operation changes either the relative order or the last character of the rotations of $w_k$.
To do so, in this part we use the notation $\vv{v}{}$ and $\vv{v}{\star}$ to denote the BWT in correspondence of the rotations with prefix $v\in \Sigma^*$ of $w_k$ and $w'_k$ respectively, where $w'_k$ is obtained after applying to $w_k$ an specific edit operation.
The number of runs in the BWT of $w_k'$ can easily be derived by comparing its BWT to the BWT of $w_k$, for which we explicitly counted the number of runs, so we omit that part of the proofs. All the edit operations on $w_k$ we show in this subsection increase the number $r(w_k)$ by a $\Theta(k)$ additive factor. To give an intuition, this increment comes mainly from the $\vv{\b^j\a}{\star}$ ranges for $2 \leq j \leq k-2$, because for each one of the corresponding ranges $\vv{\b^j\a}{} = \a\b^{2k-2j-1}\a$ in $\bwt(w_k)$, one of the $\b$'s is either moved to the top or the bottom of the range, in a consistent manner for each $j$ (it depends on the edit operation if the $\b$ goes to the top or the bottom of the range, but it is the same behavior for all the ranges considered). Then, two ranges that originally agreed on their last and first character in $w_k$ are now separated by a $\b$, adding this way $2$ new runs for each $j$. 

\bigskip
\begin{lemma}[BWT of $w_k\a$] \label{le:bwt_wk_insertion}
      Given an integer $k >5$, for $w_k\a$ it holds that
      \begin{align*}    
      \vv{\a^i\b}{\star} &= \b\a^{k-i-2} \text{ for all } 4 \leq i \leq k-2, \\
      \vv{\a^3\b}{\star} &= \b\b^5(\a\b)^{k-6}\a, \\
      \vv{\a^{2}\b}{\star} &= \a\a\a\b\a^{2k-8}, \\
      \vv{\a\b}{\star} &= \b^{k-2}\a\a\b\a^{2k-6}, \\
      \vv{\b\a}{\star} &= \a^{k-5}\b\b\b\b\a\b^{k-5}\a\b^{k-2}\a,  \\
      \vv{\b^j\a}{\star} &= \b\a\b^{2k-2j-2}\a \text{ for all } 2\leq j \leq k-1 \text{ and } \\
      \vv{\b^{k}\a}{\star} &= \a.
      \end{align*}
     Hence, $\bwt(w_k\a) = \prod_{i=2}^{k-1}\vv{\a^{k-i}\b}{\star}\cdot \prod_{i=1}^{k}\vv{\b^i\a}{\star}$. Moreover, it holds that $r(w_k\a) = 8k-20$.
\end{lemma}

\begin{proof}
    By Lemmas \ref{le:a^k-2} and \ref{le:a^ib}, we can see that appending an $\a$ after $q_k$ does not affect the BWT in the range of rotations having $\a^i\b$ as prefix, for all $4 \leq i \leq k-2$.
    Thus, $\vv{\a^i\b}{\star}=\vv{\a^i\b}{}$ for all $4 \leq i \leq k-2$.
    
    The rotation starting with $\a \a s_2$, which is not a circular factor of $w_k$, ends with a $\b$.
    By Lemma \ref{le:a^3b}, we can see that such a rotation is the smallest one with prefix $\a\a\a\b$ in lexicographic order, while the other rotations maintain their relative order.
    Therefore, $\vv{\a\a\a\b}{\star} = \b \cdot \vv{\a\a\a\b}{}$.
 
    By Lemma \ref{le:a^2b}, the rotation with prefix $\a s_2$ is still the smallest rotation starting with $\a\a\b$, with the difference that in this case, it ends with the last $\a$ of $q_k$.
    It follows that $\vv{\a\a\b}{\star}$ is obtained by replacing the first $\b$ of $\vv{\a\a\b}{}$ with an $\a$.

    Both the order and the last symbol of all the rotations having as prefix $\a\b$ described in Lemma \ref{le:ab} is not affected from the insertion of the $\a$, and therefore $\vv{\a\b}{\star} = \vv{\a\b}{}$.

    Let us now consider all the rotations of $w_k$ with prefix $\b^j \a s_2$, for all $1\leq j \leq k$.
    One can notice that $w_k \a$ does not have any rotation starting with $\b^j \a s_2$, for all $1\leq j \leq k$, but instead it has rotations starting with $\b^j \a \a s_2$.
    Thus, for all $1 \leq j \leq k-1$, to obtain $\vv{\b^j\a}{\star}$ from $\vv{\b^j\a}{}$ we have to remove the $\b$ in correspondence of the rotations starting with $\b^j\a s_2$, and add a $\b$ in correspondence of the rotations $\b^j \a \a s_2$.
    By Lemmas \ref{le:ba}, \ref{le:b^ia}, and \ref{le:b^ka}, such rotations are placed right before the rotation starting with $\b^j\a\a e_2$. 

    Finally, the last rotation has still the same prefix $\b^k \a$ and ends with an $\a$, and the thesis follows.    
\end{proof}

\bigskip 

\begin{lemma}[BWT of $\widehat{w_k}$] \label{le:bwt_wk_deletion}
      Given an integer $k >5$, for $\widehat{w_k}$ it holds that
      \begin{align*}
      \vv{\a^i\b}{\star} &= \b\a^{k-i-2} \text{ for all } 4 \leq i \leq k-2, \\ 
      \vv{\a^3\b}{\star} &= \b^5(\a\b)^{k-6}\a, \\ 
      \vv{\a^{2}\b}{\star} &= \a\a\b\a^{2k-8},\\
      \vv{\a\b}{\star} &= \b^{k-2}\b\a\b\a^{2k-6}, \\ 
      \vv{\b\a}{\star} &= \a^{k-5}\b\b\b\a\b^{k-5}\a\b^{k-2}\b\a, \\
      \vv{\b^j\a}{\star} &= \a\b^{2k-2j-2}\a\b \text{ for all } 2\leq j \leq k-1 \text{ and }\\ 
      \vv{\b^{k}\a}{\star} &= \a.
      \end{align*}
      
     Hence, $\bwt(\widehat{w_k}) = \prod_{i=2}^{k-1}\vv{\a^{k-i}\b}{\star}\cdot \prod_{i=1}^{k}\vv{\b^i\a}{\star}$. Moreover, it holds that $r(\widehat{w_k}) = 8k-20$.
\end{lemma}

\begin{proof}
    Analogously to the previous Lemma, if we look in Lemmas \ref{le:a^k-2}, \ref{le:a^ib}, and \ref{le:a^3b}, at the structure of the BWT in correspondence of the rotations starting with $\a^i\b$, for all $3 \leq i \leq k-2$, we can notice that the order or the symbols in the BWT is not affected.
    Thus, for all $3 \leq i \leq k-2$, we have $\vv{\a^i\b}{\star} = \vv{\a^i\b}{}$.

    Since the last $\a$ of $q_k$ is omitted, the circular factor $\a s_2$ does not appear anymore in $\widehat{w}$.
    Thus, $\vv{\a\a\b}{\star}$ is obtained by removing the first $\b$ from $\vv{\a\a\b}{}$, since by Lemma \ref{le:a^2b} it is in correspondence of the rotation with prefix $\a s_2$.
    
    On the other hand we can observe from Lemma \ref{le:ab} that the rotation with prefix $s_2$ maintains its relative order also in $\widehat{w_k}$, but its last symbol is now a $\b$ instead of an $\a$.

    For each $1 \leq j \leq k$, the rotation starting with $\b^j\a s_2$ of $w_k$ does not appear in $\widehat{w_k}$, but in fact it is replaced by one having $\b^j s_2$ as prefix and ending in the same way.
    When $j=1$, by Lemma \ref{le:ba} such a rotation is located between the last two rotations with the prefix $\b\a$, which start by $\b\a\b s_3$ and $\b s_3$ respectively. 
    When $2 \leq j \leq k-1$, by Lemmas \ref{le:b^ia} and \ref{le:b^ka}, the rotation starting with $\b^j s_2$ is greater than all the other rotations with prefix $\b^j\a$.
    Thus, for all $1 \leq j \leq k-1$, we obtain $\vv{\b^j\a}{\star}$ by moving the $\b$ in correspondence of the rotation starting with $\b\a s_2$ from $\vv{\b^j\a}{}$ and placing it in correspondence of $\b^j s_2$.
    Finally, the last rotation has still the same prefix $\b^k\a$ and ends with an $\a$, and the thesis follows.  
\end{proof}

\begin{lemma}[BWT of $\widehat{w_k}\b$] \label{le:bwt_wk_substitution}
      Given an integer $k >5$, for $\widehat{w_k}\b$ it holds that 
      \begin{align*}
      \vv{\a^i\b}{\star} &= \b\a^{k-i-2} \text{ for all } 4 \leq i \leq k-2, \\ 
      \vv{\a^3\b}{\star} &= \b^5(\a\b)^{k-6}\a, \\ 
      \vv{\a^{2}\b}{\star} &= \a\a\b\a^{2k-8},\\
      \vv{\a\b}{\star} &= \b^{k-2}\b\a\b\a^{2k-6}, \\ 
      \vv{\b\a}{\star} &= \a^{k-5}\b\b\b\a\b^{k-5}\a\b^{k-2}\b\a, \\
      \vv{\b^j\a}{\star} &= \a\b^{2k-2j-2}\a\b \text{ for all } 2\leq j \leq k-1,\\ 
      \vv{\b^{k}\a}{\star} &= \b \text{ and } \\
      \vv{\b^{k+1}\a}{\star} &= \a.
      \end{align*}
      
     Hence, $\bwt(\widehat{w_k}\b) = \prod_{i=2}^{k-1}\vv{\a^{k-i}\b}{\star}\cdot \prod_{i=1}^{k+1}\vv{\b^i\a}{\star}$. Moreover, it holds that $r(\widehat{w_k}\b) = 8k-20$.
\end{lemma}

\begin{proof}
    For the rotations in correspondence of the rotations starting with an $\a$, notice that replacing the last $\a$ of $w_k$ for a $\b$ or removing the last $\a$ affects the BWT in the same way.
    Therefore, $\vv{\a^i\b}{\star}$ is the same as Lemma \ref{le:bwt_wk_deletion} for all $1 \leq i \leq k-2$.

    The same behaviour can be noticed on the rotations with prefix $\b^j\a$, for all $1 \leq j \leq k-1$, while the rotation starting with $\b^k\a$ is now preceded by a $\b$.

    With respect to the other edit operations, we have the range of rotations starting with $\b^{k+1}\a$, which consists solely in $\b^{k+1}s_2 \cdots \a$. 
\end{proof}

The structure of the BWT of $w_k$ and other words obtained by applying one or more edit operations on $w_k$ are summed up in Table \ref{table:bwt_wk_edit}. 

\begin{table}[htb]
    \centering    
    \begin{adjustbox}{width=\textwidth} 
    \renewcommand{\arraystretch}{1.2}
    \begin{tabular}{|l||P{1cm}|P{1cm}|P{1cm}|P{2cm}|P{2.5cm}|P{2.5cm}|P{3cm}|}     
        \hline
          Word & $\vv{\$}{}$ & $\vv{\a\$}{}$ & $\vv{\a\a\$}{}$ & $\vv{\a^i\b}{}$ & $\vv{\a^3\b}{}$ & $\vv{\a^2\b}{}$& $\vv{\a\b}{}$ \\ \hline \hline
         
         $w_k$ & $\epsilon$ & $\epsilon$ & $\epsilon$ & $\b\a^{k-i-2}$ & $\b^5(\a\b)^{k-6}\a$ & $ \b\a\a\b\a^{2k-8}$ & $\b^{k-2}\a\a\b\a^{2k-6}$ \\ \hline
         
         $w_k\a$ & $\epsilon$ & $\epsilon$ & $\epsilon$ & $\b\a^{k-i-2}$ & $\b\b^5(\a\b)^{k-6}\a$ & $\a\a\a\b\a^{2k-8}$ & $\b^{k-2}\a\a\b\a^{2k-6}$ \\ \hline
         $\widehat{w_k}$ & $\epsilon$ & $\epsilon$ & $\epsilon$ & $\b\a^{k-i-2}$ & $ \b^5(\a\b)^{k-6}\a$ & $ \a\a\b\a^{2k-8}$ & $\b^{k-2}\b\a\b\a^{2k-6}$ \\ \hline
         $\widehat{w_k}\b$ & $\epsilon$ & $\epsilon$ & $\epsilon$ & $\b\a^{k-i-2}$ & $ \b^5(\a\b)^{k-6}\a$ & $ \a\a\b\a^{2k-8}$ & $\b^{k-2}\b\a\b\a^{2k-6}$ \\ \hline
         $w_k\$$ & $\a$ & $\b$ & $\epsilon$ & $\b\a^{k-i-2}$ & $\b^5(\a\b)^{k-6}\a$ & $\a\a\b\a^{2k-8}$ & $\b^{k-2}\$\a\b\a^{2k-6}$  \\ \hline
         $w_k\b\$$ & $\b$ & $\epsilon$ & $\epsilon$ & $\b\a^{k-i-2}$ & $\b^5(\a\b)^{k-6}\a$ & $\a\a\b\a^{2k-8}$ &  $\b\b^{k-2}\$\a\b\a^{2k-6}$\\ \hline
         $w_k\b\b\$$ & $\b$ & $\epsilon$ & $\epsilon$ & $\b\a^{k-i-2}$ & $\b^5(\a\b)^{k-6}\a$ & $\a\a\b\a^{2k-8}$ & $\b\b^{k-2}\$\a\b\a^{2k-6}$ \\ \hline
         $w_k\a\$$ & $\a$ & $\a$ & $\b$ & $\b\a^{k-i-2}$ & $\b^5(\a\b)^{k-6}\a$ & $\a\a\b\a^{2k-8}$ &  $\b^{k-2}\$\a\b\a^{2k-6}$ \\ \hline
    \end{tabular} 
    \end{adjustbox}
    \begin{adjustbox}{width=\textwidth}
    \renewcommand{\arraystretch}{1.2}
    \begin{tabular}{|l||P{1cm}|P{4cm}|P{1cm}|P{2.8cm}|P{1cm}|P{1.2cm}|P{2cm}|}
        \hline
         Word & $\vv{\b\$}{}$ & $\vv{\b\a}{}$ & $\vv{\b\b\$}{}$ & $\vv{\b^j\a}{}$& $\vv{\b^{k}\a}{}$ & $\vv{\b^{k+1}}{}$ & $r(\cdot)$ \\ \hline\hline
         $w_k$ & $\epsilon$ & $\a^{k-5}\b\b\b\a\b^{k-4}\a\b^{k-2}\a$ & $\epsilon$ &  $\a\b^{2k-2j-1}\a$ & $\a$ & $\epsilon$ & $6k-12$ \\ \hline
         $w_k\a$ & $\epsilon$ & $ \a^{k-5}\b\b\b\b\a\b^{k-5}\a\b^{k-2}\a$ & $\epsilon$ & $\b\a\b^{2k-2j-2}\a$ & $\a$ & $\epsilon$ & $8k-20$\\ \hline
         $\widehat{w_k}$ &  $\epsilon$  & $\a^{k-5}\b\b\b\a\b^{k-5}\a\b^{k-2}\b\a$ &  $\epsilon$  & $ \a\b^{2k-2j-2}\a\b$ & $\a$ & $\epsilon$ &  $8k-20$ \\ \hline
         $\widehat{w_k}\b$ &  $\epsilon$  & $\a^{k-5}\b\b\b\a\b^{k-5}\a\b^{k-2}\b\a$ &  $\epsilon$  & $ \a\b^{2k-2j-2}\a\b$ & $\b$ & $\a$ &  $8k-20$ \\ \hline
         $w_k\$$ & $\epsilon$ & $\b\a^{k-5}\b\b\b\a\b^{k-5}\a\b^{k-2}\a$ & $\epsilon$ & $\b\a\b^{2k-2j-2}\a$ & $\a$ & $\epsilon$ &  $8k-16$ \\ \hline
         $w_k\b\$$ & $\a$ & $\a^{k-5}\b\b\b\a\b^{k-5}\a\b\b^{k-2}\a$ & $\epsilon$ & $\a\b^{2k-2j-1}\a$ & $\a$ & $\epsilon$ &   $6k-13$\\ \hline
         $w_k\b\b\$$ & $\b$ & $\a^{k-5}\b\b\b\a\b^{k-5}\a\b\b^{k-2}\a$ & $\a$ & $\a\b^{2k-2j-2}\a\b$ & $\a$ & $\epsilon$ & $8k-17$ \\ \hline
         $w_k\a\$$ & $\epsilon$ & $\b\a^{k-5}\b\b\b\a\b^{k-5}\a\b^{k-2}\a$ & $\epsilon$ & $\b\a\b^{2k-2j-2}\a$ & $\a$ & $\epsilon$ & $8k-16$ \\ \hline
    \end{tabular}
    \end{adjustbox}
    \vspace{5pt}
        \caption{BWTs of the word $w_k$ and its variants after different edit operations. The word in the intersection of the column $\vv{x}{}$ with the row $w$ is the range of $\bwt(w)$ corresponding to all the rotations that have $x$ as a prefix. The columns $\vv{\a^i\b}{}$ and $\vv{\b^j\a}{}$ represent ranges of columns from $i \in [k-2, 4]$ (in that order) and $j \in [2, k-1]$, respectively. Note that the prefixes in the columns are disjoint, and cover all the possible ranges for the set of words considered. The BWT of each word is the concatenation of all the words in its row from left to right. In the last column appears the number of BWT runs of each of these words. }    \label{table:bwt_wk_edit}
\end{table}

For a given word $w \neq \epsilon$, let $w^{ins}$, $w^{del}$, and $w^{sub}$ be the words obtained by applying on $w$ an insertion, a deletion, and a substitution of a character respectively.

We compare the number of runs of $w_k$ and its variations and notice that the difference after applying one of the edit operations is $\Theta(k)$ in the three cases.

\bigskip
\begin{proposition} \label{le:w_k_sensitivity}
    There exists an infinite family of words $w$ such that: (i) $r(w^{ins}) - r(w) = \Theta(\sqrt{n})$; (ii) $r(w^{del}) - r(w) = \Theta(\sqrt{n})$; (iii) $r(w^{sub}) - r(w) = \Theta(\sqrt{n})$.
\end{proposition}\bigskip

\begin{proof} The family is composed of the words $w_k$ with $k > 5$. 
Let $n = |w_k|$. 
If $w_k^{ins} = w_k \a$, $w_k^{del} = \widehat{w_k}$, and $w_k^{sub} = \widehat{w_k} \b$, from Proposition \ref{prop:bwt_wk}, Lemma \ref{le:bwt_wk_insertion}, Lemma \ref{le:bwt_wk_deletion}, and Lemma \ref{le:bwt_wk_substitution}, we have that $r(w_k\a) = r(\widehat{w_k}) = r(\widehat{w_k}\b) = r(w_k)+(2k-8)$. From Observation \ref{le:w_k_length}, we have that $2k-8 = \Theta(\sqrt{n})$.
\end{proof}

\section{Bit catastrophes for $r_{\$}$}\label{sec:Cat_rdollar}

In this section, we discuss bit catastrophes when the parameter $r_\$$ is considered. Recall that for a word $v$, $r_\$(v) = \runs(\bwt(v\$))$. 

\subsection{When there is no bit catastrophe for $r_{\$}$} 

First let us consider the case where a symbol $c\in \Sigma$ is prepended to a word $v$.  As recently noted in \cite{AKAGI2023}, it is well known that in this case the value $r_{\$}$ can only vary by a constant value. For the sake of completeness, we include a proof.
 
\bigskip

\begin{proposition}
For any $\x\in\Sigma$, we have $r_{\$}(v)-1\leq r_{\$}(\x v) \leq r_{\$}(v)+2.$
\end{proposition}

\begin{proof}
Let us consider the list of lexicographically sorted cyclic rotations or, equivalently, the list of lexicographically sorted suffixes of $\x v\$$. (The equivalence follows from the fact that $\$$ is smaller than all other characters.) This list can be obtained from the list of suffixes of $v\$$, to which the suffix $\x v\$$ is added. Note that the relative order of all suffixes other than $\x v\$$ remains the same. Moreover, the corresponding symbols in the BWT also remain the same, except that the character $\x$ takes the place of $\$$. This replacement decreases the number of BWT-runs by $0,1$, or $2$, depending on whether this position in the BWT is preceded by a run of $\x$, followed by a run of $\x$, or both. The symbol corresponding to the new suffix $\x v\$$ (which produces the insertion of $\$$ in the corresponding position in the BWT) increases the number of BWT-runs by $1$ (if it is inserted between two existing runs), or by $2$ (if it breaks a run). 
\end{proof}

The following proposition shows that there are some cases in which $r_\$$ is not affected by any bit catastrophe.

\bigskip
\begin{proposition}\label{prop:a_rdollar}
	Let $\x$ be smaller than or equal to the smallest character in a word $v$, then $r_{\$}(v) \leq r_{\$}(v\x) \leq r_{\$}(v)+1.$
\end{proposition}

\begin{proof} 
The rotations of $v\x\$$ can be viewed as the rotation $\$v\x$, plus the rotations of $v\$$, where the occurrence of $\$$ has been replaced by $\x\$$. The smallest of these is of course $\$v\x$, since it starts with $\$$, while all others appear in the same order as before. This is because $\x$ is smaller or equal the smallest character of $v$ and greater than $\$$, and therefore, replacing $\$$ by $\x\$$ does not change the lexicograhic order of these rotations. This implies $\bwt(v\x\$) = \x\cdot\bwt(v\$)$, and thus, $r_{\$}(v) \leq r_{\$}(v\x) \leq r_{\$}(v)+1.$ 
\end{proof}

\subsection{Multiplicative bit catastrophes for $r_{\$}$}

We can derive from our results in Sec.~\ref{sec:Fibonacci} that there exist families of strings on which an edit operation can result in an increase of $r_\$$ by a multiplicative factor of $\log n$. 

\bigskip
\begin{proposition}
    Let $v$ be the Lyndon rotation of the Fibonacci word $s$ of even order $2k > 4$, and $n = |v|$. Let $v'$ be the word resulting by appending a $\b$ to $v$. Then $r_{\$}(v') = \Theta(\log n)$. 
\end{proposition}

\begin{proof}
Let $s=x_{2k}\a\b=x_{2k-1}\b\a x_{2k-2}\a\b$ be the Fibonacci word of order $2k$. One can see that 
$v= \a x_{2k}\b = \a x_{2k-2}\a\b x_{2k-1}\b$ \cite{BersteldeLuca97}. Since $v$ is a rotation of $s$, it holds that $r(v)=2$. By using Lemma \ref{le:pres_order_prep_smaller_primitive}, $r_{\$}(v)=\Theta(1)$ since $v$ is a Lyndon word. When we append $\b$ to $v$, we obtain $v'=\a x_{2k-2}\a\b x_{2k-1}\b\b$. One can note that $v'=\a x_{2k-2}\a\b x_{2k-1}\b\b$ is also a Lyndon word. Moreover, appending $\b$ to $v$ is equivalent to inserting $\b$ in $s$ at position $F_{2k-1}-2$, implying that $v'$ is a rotation of $s'$, where $s'$ is $s$ with a $\b$ inserted in position $F_{2k-1}-2$. By using Proposition \ref{prop:insert_infib}, we thus have that $r(v') = r(s') =\Theta(\log n)$. Since $v'$ is also a Lyndon word, therefore $r_{\$}(v')=\Theta(\log n)$, using Lemma \ref{le:pres_order_prep_smaller_primitive} again.
\end{proof}

\subsection{Additive bit catastrophes for $r_{\$}$}

In general, appending, deleting, or substituting with a symbol that is not the smallest of the alphabet can increase the number of runs of a word by an additive factor of $\Theta(\sqrt{n})$.  

\bigskip
\begin{lemma}[BWT of $w_k\$$] \label{le:bwt_wk_dollar}
      Given an integer $k >5$, for $w_k\$$ it holds that 
      \begin{align*}
      \vv{\$}{\star} &= \a\\
      \vv{\a\$}{\star} &= \b\\
      \vv{\a^i\b}{\star} &= \b\a^{k-i-2} \text{ for all } 4 \leq i \leq k-2, \\ 
      \vv{\a^3\b}{\star} &= \b^5(\a\b)^{k-6}\a, \\ 
      \vv{\a^{2}\b}{\star} &= \a\a\b\a^{2k-8},\\
      \vv{\a\b}{\star} &= \b^{k-2}\$\a\b\a^{2k-6}, \\
      \vv{\b\a}{\star} &= \b\a^{k-5}\b\b\b\a\b^{k-5}\a\b^{k-2}\a, \\
      \vv{\b^j\a}{\star} &= \b\a\b^{2k-2j-2}\a \text{ for all } 2\leq j \leq k-1 \text{ and }\\ 
      \vv{\b^{k}\a}{\star} &= \a.
      \end{align*}
      
     Hence, $\bwt(w_k\$) = \vv{\$}{\star} \cdot \vv{\a\$}{\star}\cdot \prod_{i=2}^{k-1}\vv{\a^{k-i}\b}{\star}\cdot \prod_{i=1}^{k}\vv{\b^i\a}{\star}$. Moreover, it holds that $r(w_k\$) = 8k-16$.
\end{lemma}

\begin{proof}
The first rotation of $\bwt(w_k\$)$ is $\$w_k$ and ends with an $\a$ because $w_k$ ends with an $\a$. Hence, $\vv{\$}{\star} = \a$. There is also a rotation $\a\$\widehat{w_k}$, which ends with a $\b$ because $\widehat{w_k}$ ends with a $\b$. Hence, $\vv{\a\$}{\star} = \b$. It lefts to compare the remaining ranges $\vv{v}{\star}$ with respect to $\vv{v}{}$. 

It is easy to see from Lemma \ref{le:a^k-2}, Lemma \ref{le:a^ib}, and Lemma \ref{le:a^3b} that $\vv{\a^i\b}{\star} = \vv{\a^i\b}{}$ for all $3 \leq i \leq k-2$.

The rotation starting with $\a s_2$ in $w_k$ does not exist anymore when $\$$ is appended to $w_k$. By Lemma \ref{le:a^2b} the remaining rotations keep their last symbols and relative order. Therefore, $\vv{\a\a\b}{\star}$ is the same as $\vv{\a\a\b}{}$ but with the first character removed, i.e., $\vv{\a\a\b}{\star}=  \a\a\b\a^{2k-8}$ .

For the rotations starting with $\a\b$, it happens that the rotation that originally started with $s_2$ in $w_k$, now ends with a $\$$. By Lemma \ref{le:ab}, the remaining rotations do not change their last symbol. Also, all the rotations keep their relative order. Hence, $\vv{\a\b}{\star} = \b^{k-2}\$\a\b\a^{2k-6}$.

In the case of the rotations starting with $\b\a$, the rotation that originally started with $\b\a s_2$ now starts with $\b\a\$s_2$ and is the smallest of its range. From Lemma \ref{le:ba} the remaining rotations keep their last symbols and relative order. Hence, $\vv{\b\a}{\star} = \b\a^{k-5}\b\b\b\a\b^{k-5}\a\b^{k-2}\a$.

For the rotations starting with $\b^j\a$ for $2 \leq j \leq k-1$, one can notice that after appending $\$$ to $w_k$, the rotation that previously started with $\b^j\a s_2$ and ended with a $\b$, now starts with $\b^j\a\$s_2$ and still ends with a $\b$. Moreover, this rotation is smaller than the rotation starting with $\b^j\a\a e_j$. From Lemma \ref{le:b^ia} and Lemma \ref{le:b^ka} we can see that all the other rotations keep their relative order and last symbols. The rotation starting with $\b^j\a\a e_j$ still ends with an $\a$, but now is the second smallest of its range. Hence, $\vv{\b^j\a}{\star} = \b\a\b^{2k-2j-2}\a$ for all $2 \leq j \leq k-1$.

Finally, it is clear that $\vv{\b^k\a}{\star} = \a$, as there is only one maximal run of $k$ symbol $\b$'s, and it is not preceded by $\$$. 
\end{proof}

\bigskip
\begin{lemma}[BWT of $w_k\b\$$]\label{le:bwt_wk_b_dollar}
      Given an integer $k >5$, for  $w_k\b\$$ it holds that
      \begin{align*}
      \vv{\$}{\star} &= \b\\
      \vv{\a^i\b}{\star} &= \b\a^{k-i-2} \text{ for all } 4 \leq i \leq k-2, \\ 
      \vv{\a^3\b}{\star} &= \b^5(\a\b)^{k-6}\a, \\ 
      \vv{\a^{2}\b}{\star} &= \a\a\b\a^{2k-8},\\
      \vv{\a\b}{\star} &= \b\b^{k-2}\$\a\b\a^{2k-6}, \\
      \vv{\b\$}{\star} &= \a, \\
      \vv{\b\a}{\star} &= \a^{k-5}\b\b\b\a\b^{k-5}\a\b\b^{k-2}\a, \\
      \vv{\b^j\a}{\star} &= \a\b^{2k-2j-1}\a \text{ for all } 2\leq j \leq k-1 \text{ and }\\ 
      \vv{\b^{k}\a}{\star} &= \a.
      \end{align*}
      
     Hence, $\bwt(w_k\b\$) = \vv{\$}{\star} \cdot (\prod_{i=2}^{k-1}\vv{\a^{k-i}\b}{\star})\cdot \vv{\b\$}{\star} \cdot (\prod_{i=1}^{k}\vv{\b^i\a}{\star})$. Moreover, it holds that $r(w_k\b\$) = 6k-13$.
\end{lemma}

\begin{proof}
The first rotation of $\bwt(w_k\b\$)$ is $\$w_k\b$. Hence, $\vv{\$}{\star} = \b$. There is also a rotation $\b\$w_k$, which ends with an $\a$ because $w_k$ ends with an $\a$. Hence, $\vv{\b\$}{\star} = \a$. It lefts to compare the remaining ranges $\vv{v}{\star}$ with respect to $\vv{v}{}$. 

It is easy to see from Lemma \ref{le:a^k-2}, Lemma \ref{le:a^ib}, and Lemma \ref{le:a^3b} that $\vv{\a^i\b}{\star} = \vv{\a^i\b}{}$ for all $3 \leq i \leq k-2$.

The rotation starting with $\a s_2$ in $w_k$ does not exist anymore when $\b\$$ is appended to $w_k$. By Lemma \ref{le:a^2b} the remaining rotations keep their last symbols and relative order. Therefore, $\vv{\a\a\b}{\star}$ is the same as $\vv{\a\a\b}{}$ but with the first character removed, i.e., $\vv{\a\a\b}{\star}=  \a\a\b\a^{2k-8}$ .

For the rotations starting with $\a\b$, it happens that the rotation that originally started with $s_2$ in $w_k$, now ends with a $\$$ when $\b\$$ is appended. Also, there is a new rotation starting with $\a\b\$$ that ends with $\b$, and is clearly the smallest of the range. By Lemma \ref{le:ab}, the remaining rotations do not change their last symbol. Also, all the rotations that come from $w_k$ keep their relative order. Hence, $\vv{\a\b}{\star} = \b\b^{k-2}\$\a\b\a^{2k-6}$.

In the case of the rotations starting with $\b\a$, the rotation that originally started with $\b\a s_2$ now starts with $\b\a\b\$s_2$ and can be found just before the rotation starting with $\b\a\b\a^{k-2}$. From Lemma \ref{le:ba} the remaining rotations keep their last symbols and relative order. Hence, $\vv{\b\a}{\star} = \a^{k-5}\b\b\b\a\b^{k-5}\a\b\b^{k-2}\a$.

For the rotations starting with $\b^j\a$ for $2 \leq j \leq k-1$, one can notice that after appending $\b\$$ to $w_k$, the rotation that previously started with $\b^j\a s_2$ and ended with a $\b$, now starts with $\b^j\a\b\$s_2$ and still ends with a $\b$. Moreover, this rotation is still strictly in between the rotations starting with $\b^j\a\a e_j$ and $\b^j\a\b\a^{j-2}s_{j+1}$ ($q_k$ instead of $s_{j+1}$ if $j = k-1$). From Lemma \ref{le:b^ia} and Lemma \ref{le:b^ka}, we can see that the latter two rotations are still the smallest and greatest of the range, and both end with an $\a$. Also, all the other rotations keep their last symbols. Hence, $\vv{\b^j\a}{\star} = \vv{\b^j\a}{}$ for all $2 \leq j \leq k-1$.

Finally, it is clear that $\vv{\b^k\a}{\star} = \a$, as there is only one maximal run of $k$ symbol $\b$'s, and it is not preceded by $\$$. 
\end{proof}

\bigskip
\begin{lemma}[BWT of $w_k\b\b\$$] \label{le:bwt_wk_b_b_dollar}
      Given an integer $k >5$, for $w_k\b\b\$$ it holds that 
      \begin{align*}
      \vv{\$}{\star} &= \b\\
      \vv{\a^i\b}{\star} &= \b\a^{k-i-2} \text{ for all } 4 \leq i \leq k-2, \\ 
      \vv{\a^3\b}{\star} &= \b^5(\a\b)^{k-6}\a, \\ 
      \vv{\a^{2}\b}{\star} &= \a\a\b\a^{2k-8},\\
      \vv{\a\b}{\star} &= \b\b^{k-2}\$\a\b\a^{2k-6}, \\
      \vv{\b\$}{\star} &= \b, \\
      \vv{\b\a}{\star} &= \a^{k-5}\b\b\b\a\b^{k-5}\a\b\b^{k-2}\a, \\
      \vv{\b\b\$}{\star} &= \a, \\
      \vv{\b^j\a}{\star} &= \a\b^{2k-2j-2}\a\b \text{ for all } 2\leq j \leq k-1 \text{ and }\\ 
      \vv{\b^{k}\a}{\star} &= \a.
      \end{align*}
      
     Hence, $\bwt(w_k\b\b\$) = \vv{\$}{\star} \cdot (\prod_{i=2}^{k-1}\vv{\a^{k-i}\b}{\star})\cdot \vv{\b\$}{\star} \cdot \vv{\b\a}{\star} \cdot \vv{\b\b\$}{\star} \cdot (\prod_{i=2}^{k}\vv{\b^i\a}{\star})$. Moreover, it holds that $r(w_k\b\$) = 8k-17$.
\end{lemma}

\begin{proof}
The first rotation of $\bwt(w_k\b\b\$)$ is $\$w_k\b\b$. Hence, $\vv{\$}{\star} = \b$. There is another new rotation $\b\$w_k\b$. Hence, $\vv{\b\$}{\star} = \b$. There is also a rotation $\b\b\$w_k$ that ends with an $\a$ because $w_k$ ends with an $\a$. Hence, $\vv{\b\b\$}{\star} = \a$. It lefts to compare the remaining ranges $\vv{v}{\star}$ with respect to $\vv{v}{}$. 

It is easy to see from Lemma \ref{le:a^k-2}, Lemma \ref{le:a^ib}, and Lemma \ref{le:a^3b} that $\vv{\a^i\b}{\star} = \vv{\a^i\b}{}$ for all $3 \leq i \leq k-2$.

The rotation starting with $\a s_2$ in $w_k$ does not exist anymore when $\b\b\$$ is appended to $w_k$. By Lemma \ref{le:a^2b} the remaining rotations keep their last symbols and relative order. Therefore, $\vv{\a\a\b}{\star}$ is the same as $\vv{\a\a\b}{}$ but with the first character removed, i.e., $\vv{\a\a\b}{\star}=  \a\a\b\a^{2k-8}$ .

For the rotations starting with $\a\b$, it happens that the rotation that originally started with $s_2$ in $w_k$, now ends with a $\$$ when $\b\b\$$ is appended. Also, there is a new rotation starting with $\a\b\b\$$ that ends with $\b$, and can be found just before the rotation starting with $s_2$. By Lemma \ref{le:ab}, the remaining rotations do not change their last symbol. Also, all the rotations that come from $w_k$ keep their relative order. Hence, $\vv{\a\b}{\star} = \b^{k-2}\b\$\a\b\a^{2k-6}$.

In the case of the rotations starting with $\b\a$, the rotation that originally started with $\b\a s_2$ now starts with $\b\a\b\b\$s_2$ and can be found just before the rotation starting with $\b s_3$ (the greatest on the range). From Lemma \ref{le:ba} we can see that the remaining rotations keep their last symbols and relative order. Hence, $\vv{\b\a}{\star} = \a^{k-5}\b\b\b\a\b^{k-5}\a\b^{k-2}\b\a$.

For the rotations starting with $\b^j\a$ for $2 \leq j \leq k-1$, one can notice that after appending $\b\b\$$ to $w_k$, the rotation that previously started with $\b^j\a s_2$ and ended with a $\b$, now starts with $\b^j\a\b\b\$s_2$ and still ends with a $\b$. Moreover, this rotation is greater than the rotation starting with $\b^j\a\b\a^{j-2}s_{j+1}$ ($q_k$ instead of $s_{j+1}$ if $j = k-1$).  From Lemma \ref{le:b^ia} and Lemma \ref{le:b^ka} we can see that all the other rotations keep their relative order an last symbols. The rotation starting with  $\b^j\a\b\a^{j-2}s_{j+1}$ ($q_k$ instead of $s_{j+1}$ if $j = k-1$) still ends with an $\a$, but now is the second greatest of its range. Hence, $\vv{\b^j\a}{\star} = \a\b^{2k-2j-2}\a\b$ for all $2 \leq j \leq k-1$.

Finally, it is clear that $\vv{\b^k\a}{\star} = \a$, as there is only one maximal run of $k$ symbol $\b$'s, and it is not preceded by $\$$. 
\end{proof}

\bigskip
\begin{lemma}[BWT of $w_k\a\$$] \label{le:bwt_wk_a_dollar}
      Given an integer $k >5$, for $w_k\a\$$ it holds that
      \begin{align*}
      \vv{\$}{\star} &= \a\\
       \vv{\a\$}{\star} &= \a\\
        \vv{\a\a\$}{\star} &= \b\\
      \vv{\a^i\b}{\star} &= \b\a^{k-i-2} \text{ for all } 4 \leq i \leq k-2, \\ 
      \vv{\a^3\b}{\star} &= \b^5(\a\b)^{k-6}\a, \\ 
      \vv{\a^{2}\b}{\star} &= \a\a\b\a^{2k-8},\\
      \vv{\a\b}{\star} &= \b^{k-2}\$\a\b\a^{2k-6}, \\
      \vv{\b\a}{\star} &= \b\a^{k-5}\b\b\b\a\b^{k-5}\a\b^{k-2}\a, \\
      \vv{\b^j\a}{\star} &= \b\a\b^{2k-2j-2}\a \text{ for all } 2\leq j \leq k-1 \text{ and }\\ 
      \vv{\b^{k}\a}{\star} &= \a.
      \end{align*}
      
     Hence, $\bwt(w_k\a\$) = \vv{\$}{\star} \cdot (\prod_{i=2}^{k-1}\vv{\a^{k-i}\b}{\star})\cdot \vv{\b\$}{\star} \cdot (\prod_{i=1}^{k}\vv{\b^i\a}{\star})$. Moreover, it holds that $r(w_k\a\$) = 8k-16$.
\end{lemma}

\begin{proof}
We obtain $\bwt(w_k\a\$) = \a\bwt(w_k\$)$ by applying Proposition \ref{prop:a_rdollar} to the words $w_k\a\$$ and $w_k\$$, and we already know the structure of $\bwt(w_k\$)$ by Lemma \ref{le:bwt_wk_dollar}. 
\end{proof}

\bigskip
\begin{proposition} \label{le:w_k_b_sensitivity_r_dollar}
    There exists an infinite family of words such that: (i) $r_\$(w\b) - r_\$(w) = \Theta(\sqrt{n})$; (ii) $r_\$(\widehat{w}) - r_\$(w) = \Theta(\sqrt{n})$; (iii) $r_\$(\widehat{w}\a) - r_\$(w) = \Theta(\sqrt{n})$.
\end{proposition}

\begin{proof} Such a family is composed of the words $w_k\b$ with $k > 5$. The proof follows from Lemma \ref{le:bwt_wk_dollar}, Lemma \ref{le:bwt_wk_b_dollar}, Lemma \ref{le:bwt_wk_b_b_dollar},  Lemma \ref{le:bwt_wk_a_dollar}, and Observation \ref{le:w_k_length}. 
\end{proof}

\subsection{The relationship between $r$ and $r_{\$}$} 

Now we address the differences between the measures $r$ and $r_\$$. In fact, not only are the measures $r$ and $r_\$$ not equal over the same input, but they may differ by a $\Theta(\log{n})$ multiplicative factor, or by a $\Theta(\sqrt{n})$ additive factor.

\bigskip
\begin{proposition}
    There exists an infinite family of words $v$ such that $r_{\$}(v)/r(v)=\Theta(\log n)$, where $n=|v|$. 
\end{proposition}

\begin{proof}
    The family consists of the reverse of the Fibonacci words of odd order. Let $v = \rev(s)$, with $s$ a Fibonacci word of odd order $2k+1$. Since $s$ is a standard word, $r(s)=2$. Moreover, its reverse $v$ is a conjugate and thus $\bwt(v)=\bwt(s)$, implying that also $r(v)=2$. Let $v' = v\$$. Since $\$ < \a$, by 
    Proposition~\ref{prop:furtherC_notinSigma_odd} it follows that $r(v') \in \{2k+2, 2k+3\}$. Altogether, $r_{\$}(v)/r(v) \leq \frac{2k+3}{2} = \Theta(k) = \Theta(\log n)$. 
\end{proof}

\begin{proposition}
    There exists an infinite family of words $w$ such that $r_{\$}(w)-r(w)=\Theta(\sqrt{n})$, where $n = |w|$.
\end{proposition}

\begin{proof} 
    The family consists of the words $w_k$ for all $k > 5$, defined in Section \ref{sec:additive_sqrt_n}.
    From Proposition \ref{prop:bwt_wk} and Lemma \ref{le:bwt_wk_dollar}, it holds $r_{\$}(w_k)-r(w_k) =  2k-4$. By Observation \ref{le:w_k_length}, it holds $r_{\$}(w_k)-r(w_k) = \Theta(\sqrt{n})$. 
\end{proof}

\section{Conclusion}\label{sec:conclusion} 

In this paper, we studied how a single edit operation on a word (insertion, deletion or substitution of a character) can affect the number of runs $r$ of the BWT of the word. Our contribution is threefold. First, we prove that $\Omega(\log n)$ is a lower bound for all three edit operations, by exhibiting infinite families of words for which each edit operation can increase the number of runs by a multiplicative $\Theta(\log n)$ factor. Since for all of these families, $r = \Oh(1)$, this also proves that the upper bound $\Oh(\log n \log r)$ given in~\cite{AKAGI2023} is tight in the case of $r=\Oh(1)$, for each of the three edit operations. Secondly, we improved the best known lower bound of $\Omega(\log n)$ for the additive sensitivity of $r$~\cite{GILPST21,AKAGI2023}, by giving an infinite family of words on which insertion, deletion, and substitution of a character can increase $r$ by a $\Theta(\sqrt{n})$ additive factor. Finally, we put in relation the two common variants of the number of runs of the BWT, which we denote as $r$ resp.\ $r_\$$. The latter, $r_\$$, is the variant used in articles on string data structures and compression, which assumes that each word is terminated by an end-of-string symbol; for the variant $r$ commonly used in the literature on combinatorics on words, no such assumption is made.

Our work opens several roads of investigation. First, we ask whether there exist families of words with $r=\omega(1)$ for which edit operations can cause a multiplicative increase of $\Omega(\log n)$. In other words, is the bit catastrophe effect restricted to words on which the compression power of $r$ is maximal? 

Another interesting question is whether the upper bound $\Oh(r \log r \log n)$ from~\cite{AKAGI2023} for the additive sensitivity of $r$ is tight. A weaker question, an answer to which would make a step in this direction, is whether there exists an infinite family with $r = \omega(1)$ on which one edit operation can cause an additive increase of $\omega(r)$ in the number of runs of the BWT. 

\subsection*{Funding}

CU is funded by scholarship ANID-Subdirección de Capital Humano/Doctorado Nacional/2021-21210580, ANID, Chile. ZsL, GR, and MS are partially funded by the MUR PRIN Project \vir{PINC, Pangenome INformatiCs: from Theory to Applications} (Grant No.\ 2022YRB97K), and by the INdAM - GNCS Project CUP$\_$E53C23001670001. 
SI is partially funded by JSPS KAKENHI grant numbers JP20H05964, JP23K24808, and JP23K18466.

\subsection*{Authors contributions}

All authors contributed equally to the paper. 

\bibliography{refs}

\begin{thebibliography}{}
\providecommand{\doi}[1]{\url{https://doi.org/#1}}
\bibcommenthead

\bibitem[\protect\citeauthoryear{Akagi, Funakoshi, and Inenaga}{Akagi et~al.}{2023}]{AKAGI2023}
Akagi, T., M.~Funakoshi, and S.~Inenaga. 2023.
\newblock Sensitivity of string compressors and repetitiveness measures.
\newblock {\em Information and Computation\/}~291: 104999 .

\bibitem[\protect\citeauthoryear{Bannai, Gagie, and I}{Bannai et~al.}{2020}]{BannaiGI20}
Bannai, H., T.~Gagie, and T.~I. 2020.
\newblock Refining the \emph{r}-index.
\newblock {\em Theor. Comput. Sci.\/}~812: 96--108 .

\bibitem[\protect\citeauthoryear{Berstel and de~Luca}{Berstel and de~Luca}{1997}]{BersteldeLuca97}
Berstel, J. and A.~de~Luca. 1997.
\newblock {Sturmian} words, {Lyndon} words and trees.
\newblock {\em Theoretical Computer Science\/}~{\em 178\/}(1-2): 171--203 .

\bibitem[\protect\citeauthoryear{Boucher, Cenzato, Lipt{\'{a}}k, Rossi, and Sciortino}{Boucher et~al.}{2021}]{BoucherCL0S21}
Boucher, C., D.~Cenzato, {\relax Zs}.~Lipt{\'{a}}k, M.~Rossi, and M.~Sciortino 2021.
\newblock $r$-indexing the {eBWT}.
\newblock In T.~Lecroq and H.~Touzet (Eds.), {\em Proc.\ of 28th International Symposium on String Processing and Information Retrieval ({SPIRE} 2021)}, Volume 12944 of {\em Lecture Notes in Computer Science}, pp.\  3--12. Springer.

\bibitem[\protect\citeauthoryear{Boucher, Cenzato, Lipt{\'{a}}k, Rossi, and Sciortino}{Boucher et~al.}{2024}]{BoucherCL0S24}
Boucher, C., D.~Cenzato, {\relax Zs}.~Lipt{\'{a}}k, M.~Rossi, and M.~Sciortino. 2024.
\newblock $r$-indexing the {eBWT}.
\newblock {\em Information and Computation\/}~298: 105155.
\newblock \doi{10.1016/j.ic.2024.105155} .

\bibitem[\protect\citeauthoryear{Brlek, Frosini, Mancini, Pergola, and Rinaldi}{Brlek et~al.}{2019}]{BrlekFMPR19}
Brlek, S., A.~Frosini, I.~Mancini, E.~Pergola, and S.~Rinaldi 2019.
\newblock {Burrows-Wheeler Transform of Words Defined by Morphisms}.
\newblock In {\em {IWOCA}}, Volume 11638 of {\em Lect. Notes Comput. Sci.}, pp.\  393--404. Springer.

\bibitem[\protect\citeauthoryear{Burrows and Wheeler}{Burrows and Wheeler}{1994}]{BurrowsWheeler94}
Burrows, M. and D.J. Wheeler 1994.
\newblock A block-sorting lossless data compression algorithm.
\newblock Technical report, DIGITAL System Research Center.

\bibitem[\protect\citeauthoryear{Castiglione, Restivo, and Sciortino}{Castiglione et~al.}{2010}]{CS2010}
Castiglione, G., A.~Restivo, and M.~Sciortino. 2010.
\newblock On extremal cases of {H}opcroft's algorithm.
\newblock {\em Theoret. Comput. Sci.\/}~{\em 411\/}(38-39): 3414--3422 .

\bibitem[\protect\citeauthoryear{de~Luca}{de~Luca}{1997}]{deLuca97a}
de~Luca, A. 1997.
\newblock {Sturmian} words: Structure, combinatorics, and their arithmetics.
\newblock {\em Theor. Comput. Sci.\/}~{\em 183\/}(1): 45--82 .

\bibitem[\protect\citeauthoryear{de~Luca and Mignosi}{de~Luca and Mignosi}{1994}]{deLucaM94}
de~Luca, A. and F.~Mignosi. 1994.
\newblock {Some Combinatorial Properties of {S}turmian Words}.
\newblock {\em Theor. Comput. Sci.\/}~{\em 136\/}(2): 361--285 .

\bibitem[\protect\citeauthoryear{Ferragina and Manzini}{Ferragina and Manzini}{2000}]{FerraginaM00}
Ferragina, P. and G.~Manzini 2000.
\newblock Opportunistic data structures with applications.
\newblock In {\em {FOCS}}, pp.\  390--398. {IEEE} Computer Society.

\bibitem[\protect\citeauthoryear{Frosini, Mancini, Rinaldi, Romana, and Sciortino}{Frosini et~al.}{2022}]{FrosiniMRRS22_DLT}
Frosini, A., I.~Mancini, S.~Rinaldi, G.~Romana, and M.~Sciortino 2022.
\newblock Logarithmic equal-letter runs for {BWT} of purely morphic words.
\newblock In {\em {DLT}}, Volume 13257 of {\em Lect. Notes Comput. Sci.}, pp.\  139--151. Springer.

\bibitem[\protect\citeauthoryear{Gagie, Navarro, and Prezza}{Gagie et~al.}{2018}]{GagieNP18}
Gagie, T., G.~Navarro, and N.~Prezza 2018.
\newblock Optimal-time text indexing in {BWT}-runs bounded space.
\newblock In A.~Czumaj (Ed.), {\em {SODA}}, pp.\  1459--1477. {SIAM}.

\bibitem[\protect\citeauthoryear{Gagie, Navarro, and Prezza}{Gagie et~al.}{2020}]{GagieNP20}
Gagie, T., G.~Navarro, and N.~Prezza. 2020.
\newblock Fully functional suffix trees and optimal text searching in {BWT}-runs bounded space.
\newblock {\em J. {ACM}\/}~{\em 67\/}(1): 2:1--2:54 .

\bibitem[\protect\citeauthoryear{Giuliani, Inenaga, Lipt{\'{a}}k, Prezza, Sciortino, and Toffanello}{Giuliani et~al.}{2021}]{GILPST21}
Giuliani, S., S.~Inenaga, {\relax Zs}.~Lipt{\'{a}}k, N.~Prezza, M.~Sciortino, and A.~Toffanello 2021.
\newblock Novel results on the number of runs of the {Burrows-Wheeler-Transform}.
\newblock In {\em {SOFSEM}}, Volume 12607 of {\em LNCS}, pp.\  249--262. Springer.

\bibitem[\protect\citeauthoryear{Giuliani, Inenaga, Lipt{\'a}k, Romana, Sciortino, and Urbina}{Giuliani et~al.}{2023}]{GILRSU_DLT}
Giuliani, S., S.~Inenaga, {\relax Zs}.~Lipt{\'a}k, G.~Romana, M.~Sciortino, and C.~Urbina 2023.
\newblock Bit catastrophes for the {Burrows-Wheeler Transform}.
\newblock In F.~Drewes and M.~Volkov (Eds.), {\em Developments in Language Theory (DLT 2021)}, Cham, pp.\  86--99. Springer Nature Switzerland.

\bibitem[\protect\citeauthoryear{Kempa and Kociumaka}{Kempa and Kociumaka}{2022}]{KempaK22}
Kempa, D. and T.~Kociumaka. 2022.
\newblock Resolution of the {Burrows-Wheeler Transform} conjecture.
\newblock {\em Commun. {ACM}\/}~{\em 65\/}(6): 91--98 .

\bibitem[\protect\citeauthoryear{Knuth, Morris, and Pratt}{Knuth et~al.}{1977}]{KMP77}
Knuth, D.E., J.H. Morris, and V.R. Pratt. 1977.
\newblock Fast pattern matching in strings.
\newblock {\em SIAM J. Comput.\/}~{\em 6\/}(2): 323--350 .

\bibitem[\protect\citeauthoryear{Kociumaka, Navarro, and Prezza}{Kociumaka et~al.}{2020}]{KociumakaNP20}
Kociumaka, T., G.~Navarro, and N.~Prezza 2020.
\newblock Towards a definitive measure of repetitiveness.
\newblock In {\em {LATIN}}, Volume 12118 of {\em Lect. Notes Comput. Sci.}, pp.\  207--219. Springer.

\bibitem[\protect\citeauthoryear{Lagarde and Perifel}{Lagarde and Perifel}{2018}]{LZbitcatastrophe}
Lagarde, G. and S.~Perifel 2018.
\newblock {Lempel-Ziv}: a ``one-bit catastrophe" but not a tragedy.
\newblock In {\em {SODA}}, pp.\  1478--1495. {SIAM}.

\bibitem[\protect\citeauthoryear{{Lam}, {Li}, {Tam}, {Wong}, {Wu}, and {Yiu}}{{Lam} et~al.}{2009}]{soap2}
{Lam}, T.W., R.~{Li}, A.~{Tam}, S.~{Wong}, E.~{Wu}, and S.M. {Yiu} 2009.
\newblock {High Throughput Short Read Alignment via Bi-directional BWT}.
\newblock In {\em {BIBM}}, pp.\  31--36. {IEEE} Computer Society.

\bibitem[\protect\citeauthoryear{Langmead, Trapnell, Pop, and Salzberg}{Langmead et~al.}{2009}]{Bowtie}
Langmead, B., C.~Trapnell, M.~Pop, and S.L. Salzberg. 2009.
\newblock Ultrafast and memory-efficient alignment of short {DNA} sequences to the human genome.
\newblock {\em Genome Biology\/}~{\em 10\/}(3): R25 .

\bibitem[\protect\citeauthoryear{Li and Durbin}{Li and Durbin}{2010}]{bwa}
Li, H. and R.~Durbin. 2010.
\newblock {Fast and accurate long-read alignment with Burrows–Wheeler transform}.
\newblock {\em Bioinformatics\/}~{\em 26\/}(5): 589--595 .

\bibitem[\protect\citeauthoryear{Lothaire}{Lothaire}{2002}]{Loth2}
Lothaire, M. 2002.
\newblock {\em Algebraic Combinatorics on {W}ords}.
\newblock Cambridge Univ.\ Press.

\bibitem[\protect\citeauthoryear{M{\"{a}}kinen and Navarro}{M{\"{a}}kinen and Navarro}{2005}]{MakinenN05}
M{\"{a}}kinen, V. and G.~Navarro 2005.
\newblock Succinct suffix arrays based on run-length encoding.
\newblock In {\em {CPM}}, Volume 3537 of {\em Lecture Notes in Comp.\ Sc.}, pp.\  45--56. Springer.

\bibitem[\protect\citeauthoryear{Mantaci, Restivo, Rosone, Sciortino, and Versari}{Mantaci et~al.}{2017}]{MANTACI_TCS2017}
Mantaci, S., A.~Restivo, G.~Rosone, M.~Sciortino, and L.~Versari. 2017.
\newblock {Measuring the clustering effect of BWT via RLE}.
\newblock {\em Theoret. Comput. Sci.\/}~698: 79 -- 87 .

\bibitem[\protect\citeauthoryear{Mantaci, Restivo, and Sciortino}{Mantaci et~al.}{2003}]{MantaciRS03}
Mantaci, S., A.~Restivo, and M.~Sciortino. 2003.
\newblock {Burrows--Wheeler} transform and {Sturmian} words.
\newblock {\em Information Processing Letters\/}~{\em 86\/}(5): 241--246 .

\bibitem[\protect\citeauthoryear{Navarro}{Navarro}{2022}]{Navarro21a}
Navarro, G. 2022.
\newblock Indexing highly repetitive string collections, part {I:} {Repetitiveness} measures.
\newblock {\em {ACM} Comput. Surv.\/}~{\em 54\/}(2): 29:1--29:31 .

\bibitem[\protect\citeauthoryear{Rosone and Sciortino}{Rosone and Sciortino}{2013}]{RosoneS13}
Rosone, G. and M.~Sciortino 2013.
\newblock The {Burrows-Wheeler} transform between data compression and combinatorics on words.
\newblock In {\em 9th Conference on Computability in Europe (CiE 2013)}, pp.\  353--364.

\bibitem[\protect\citeauthoryear{Sciortino and Zamboni}{Sciortino and Zamboni}{2007}]{SciortinoZ07}
Sciortino, M. and L.Q. Zamboni 2007.
\newblock Suffix automata and standard {Sturmian} words.
\newblock In {\em Developments in Language Theory (DLT 2007)}, Volume 4588 of {\em Lect. Notes Comput. Sc.}, pp.\  382--398. Springer.

\bibitem[\protect\citeauthoryear{Seward}{Seward}{1996}]{bzip}
Seward, J. 1996.
\newblock https://sourceware.org/bzip2/manual/manual.html.

\bibitem[\protect\citeauthoryear{Ziv and Lempel}{Ziv and Lempel}{1977}]{ZivL77}
Ziv, J. and A.~Lempel. 1977.
\newblock A universal algorithm for sequential data compression.
\newblock {\em {IEEE} Trans. Inf. Theory\/}~{\em 23\/}(3): 337--343 .

\bibitem[\protect\citeauthoryear{Ziv and Lempel}{Ziv and Lempel}{1978}]{ZivL78}
Ziv, J. and A.~Lempel. 1978.
\newblock Compression of individual sequences via variable-rate coding.
\newblock {\em {IEEE} Trans. Inf. Theory\/}~{\em 24\/}(5): 530--536 .

\end{thebibliography}

\backmatter

\end{document}